\documentclass[showpacs,onecolumn,aps,pra,longbibliography,superscriptaddress,notitlepage]{revtex4-2}
\usepackage{qcircuit}
\usepackage[dvips]{graphicx}
\usepackage{amsmath,amssymb,amsthm,mathrsfs,amsfonts,dsfont}
\usepackage{subfigure, epsfig}
\usepackage{braket}
\usepackage{svg}
\usepackage{bm}
\usepackage{enumerate}
\usepackage{algorithm}
\usepackage{algpseudocode}
\usepackage{diagbox}
\usepackage{physics}
\usepackage{color}
\usepackage{multirow}
\usepackage[marginal]{footmisc}
\usepackage{comment}
\usepackage{tikz}
\usepackage[]{qcircuit}
\usepackage{makecell}
\usetikzlibrary{arrows}
\usetikzlibrary{shapes,fadings,snakes}
\usetikzlibrary{decorations.pathmorphing,patterns}
\usetikzlibrary{calc}
\usetikzlibrary{positioning}
\usepackage[colorlinks = true]{hyperref}
\usepackage[nameinlink,capitalize]{cleveref}

\graphicspath{{./figure/}}

\renewcommand{\sec}[1]{\hyperref[sec:#1]{Section~\ref*{sec:#1}}}
\newcommand{\app}[1]{\hyperref[app:#1]{Appendix~\ref*{app:#1}}}
\newcommand{\thm}[1]{\hyperref[thm:#1]{Theorem~\ref*{thm:#1}}}
\newcommand{\lem}[1]{\hyperref[lem:#1]{Lemma~\ref*{lem:#1}}}
\newcommand{\cor}[1]{\hyperref[cor:#1]{Corollary~\ref*{cor:#1}}}
\newcommand{\fgr}[1]{\hyperref[fgr:#1]{Figure~\ref*{fgr:#1}}}
\newcommand{\tab}[1]{\hyperref[tab:#1]{Table~\ref*{tab:#1}}}

\newtheorem{assumption}{Assumption}
\newtheorem{theorem}{Theorem}

\newtheorem{lemma}{Lemma}
\newtheorem{corollary}{Corollary}
\newtheorem{proposition}{Proposition}

\newtheorem{definition}{Definition}

\newcommand{\mc}{\mathcal}

\newcommand{\comments}[1]{}

\usepackage{enumitem}

\begin{document}
\title{Infinite-dimensional Extension of the Linear Combination of Hamiltonian Simulation: Theorems and Applications}

\begin{abstract}
We generalize the Linear Combination of Hamiltonian Simulation (LCHS) formula  [An, Liu, Lin, Phys. Rev. Lett. 2023] to simulate time-evolution operators in infinite-dimensional spaces, including scenarios involving unbounded operators. This extension, named Inf-LCHS for short, bridges the gap between finite-dimensional quantum simulations and the
broader class of infinite-dimensional quantum dynamics governed by partial differential equations (PDEs). Furthermore, we propose two sampling methods by integrating the infinite-dimensional LCHS with Gaussian quadrature schemes (Inf-LCHS-Gaussian) or Monte Carlo integration schemes (Inf-LCHS-MC). We demonstrate the applicability of the Inf-LCHS theorem to a wide range of non-Hermitian dynamics, including linear parabolic PDEs, queueing models (birth-or-death processes), Schrödinger equations with complex potentials, Lindblad equations, and black hole thermal field equations. Our analysis provides insights into simulating general linear dynamics using a finite number of quantum dynamics and includes cost estimates for the corresponding quantum algorithms. 
\end{abstract}


\author{Rundi Lu}
\affiliation{Yau Mathematical Sciences Center and Department of Mathematical Sciences, Tsinghua University, Beijing 100084, China}

\author{Hao-En Li}
\affiliation{Department of Chemistry, Tsinghua University, Beijing 100084, China}

\author{Zhengwei Liu}
\email{liuzhengwei@tsinghua.edu.cn}
\affiliation{Yau Mathematical Sciences Center and Department of Mathematical Sciences, Tsinghua University, Beijing 100084, China}
\affiliation{Yanqi Lake Beijing Institute of Mathematical Sciences and Applications, Beijing 101408, China}

\author{Jin-Peng Liu}
\email{liujinpeng@tsinghua.edu.cn}
\affiliation{Yau Mathematical Sciences Center and Department of Mathematical Sciences, Tsinghua University, Beijing 100084, China}
\affiliation{Yanqi Lake Beijing Institute of Mathematical Sciences and Applications, Beijing 101408, China}

\maketitle

\tableofcontents

\newpage

\section{Introduction}

The time-evolution operators or propagators are ubiquitously used to describe the change of the states brought about by the passage of time, applicable to various systems across from mathematics, physics, to applied sciences. In a closed quantum system, the time-evolution operator is unitary on a Hilbert space, and it can be expressed as time-ordered exponentials of the integrated Hamiltonian. In an open quantum system or a classical system, the non-unitary propagator operates the phase state of a physical system. 
The simulation of time-evolution operators is closely related to the problem of solving partial differential equations (PDEs), such as the Schrödinger equation, the Lindblad equation, and the parabolic PDEs such as the heat equation, Fokker-Planck equation, advection-diffusion equation, etc.  
The efforts to understand and simplify the time-evolution operators are of great significance for solving the corresponding PDEs, either in an analytical way or using a numerical method.
The time-evolution operator is also a fundamental tool in quantum field theory, governing system dynamics in scattering, curved spacetime, and strong-field phenomena like Hawking radiation. 
Understanding and efficiently simulating time-evolution operators not only advances the study of fundamental physics but also enables practical applications in numerical analysis, quantum computing, statistical mechanics, and quantum field theory, making them a crucial subject of research across multiple disciplines.

Recent advancements in quantum computing have led to an increasing number of studies focused on simulating time-evolution processes on quantum computers. As the primary usage of quantum computers, Hamiltonian simulation algorithms have been extensively developed for simulating unitary quantum dynamics~\cite{Wie96,Zal98,KWB17,BAC07,PQS11,BCC13,BCC15,BCK15,LC16,LC17,LW18,CMN17,COS19,BCS20,CST21,SBW21,SHC21,FangLiuSarkar2024,AFL20,AFL22,ZZS21,CCH22,CLLL22}. Such simulation-based algorithms can be generalized to simulate generic non-unitary processes based on \emph{Schrödingerisation}~\cite{JLY24,JLY23,JLLY23,JLLY24,HJL23,HJZ23,HJLZ24} or \emph{Linear Combination of Hamiltonian Simulation} (LCHS) \cite{ALL23,ACL23,ACLY24} instead of complicated quantum-linear-system-based algorithms~\cite{HHL08,Ber14,BCOW17,CL19,CLO20}. The Schrödingerisation can be viewed as a ``continuing block-encoding'' of non-unitary operators. This approach offers a constructive example of Nagy's unitary dilation theorem via dilating operators of the form $V(t) = e^{-A t}$ into a unitary operator~\cite{HJL23}. 
As proposed in a similar time, the alternative LCHS approach can express the non-unitary dynamics as a linear combination of Hamiltonian simulation problems~\cite{ALL23,ACL23}. The proof of the LCHS formula does not rely on dilations or the spectral mapping theorem, which generalizes the range of matrix-function-based quantum algorithms like quantum signal processing (QSP)~\cite{LC16,LC17} and quantum singular value transformation (QSVT)~\cite{GSLW18}. Moreover, the LCHS formula does not rely on any commutators of the Hermitian part $L(t)$ and the anti-Hermitian part $H(t)$ of the time-dependent matrix $A(t) = L(t) + \mathrm{i}H(t)$. However, the theorem and proof of the LCHS equality focus on the finite-dimensional time-evolution operators from ODEs. For general time-evolution operators from PDEs, we need to first discretize the PDE in time before leveraging LCHS.  However, directly approximating PDEs from Schrödinger equations with complex and continuous spectra remains challenging. This is due to unbounded operators not easily accommodating norm-based analysis techniques and the need to remove dimensionality restrictions. Henceforth, the strict assumptions and application range of such an equality are still less understood. This motivates the mathematical extension of the LCHS formula to infinite-dimensional unbounded operators. 

In this work, we generalize the LCHS formula to simulate time-evolution operators in infinite-dimensional spaces, including scenarios involving unbounded operators, by leveraging the mathematical properties of semigroups of such operators. 
The extended formula expresses the non-unitary time evolution as a linear combination of unitary time-evolution operators. 
The \emph{infinite-dimensional extension of the linear combination of Hamiltonian simulation}, named Inf-LCHS for short, bridges the gap between finite-dimensional quantum simulations and the broader class of infinite-dimensional quantum dynamics governed by PDEs. This work aims to provide a step toward a more comprehensive framework for analyzing and approximating non-unitary dynamics in settings that inherently involve infinite-dimensional Hilbert spaces. Two conditions are provided to guarantee the validity of the Inf-LCHS theorem for unbounded operators: one ensures the existence of the BCH series for time-integrated operators of the Hamiltonian's real and imaginary parts, and the other ensures the boundedness of the time-integrated operator for the real part of the Hamiltonian. Under the assumption of a finite bound of the time-integrated operator for the real part of the evolution operator, the Inf-LCHS theorem avoids dependence on commutators between the Hermitian and anti-Hermitian parts of the operator. To decompose general time-evolution operators with a certain number and form of unitary operators, we propose two sampling methods by integrating Inf-LCHS with Gaussian quadrature schemes (Inf-LCHS-Gaussian) or Monte Carlo integration schemes (Inf-LCHS-MC). Inf-LCHS-Gaussian relies on the norm of the operator $L(t)$, while Inf-LCHS-MC can avoid such a dependence and hence be in favor of unbounded operators like parabolic PDEs and Schrödinger equations. 

We demonstrate the applicability of the Inf-LCHS theorem to a wide range of non-Hermitian dynamics, including examples such as linear parabolic PDEs, queueing models (birth-or-death processes), Schrödinger equations with complex potentials, Lindblad equations, and black hole thermal field equations. For those applications, our analysis provides insights into simulating general linear dynamics using a finite number of quantum dynamics and includes cost estimates for the corresponding quantum algorithms.

For the rest of the paper: \sec{results} presents the main theorem and assumptions; \sec{proof} gives the proof of the main theorem; \sec{sample} proposes the implementation of algorithms and sampling methods; \sec{applications} discusses several applications of simulating non-Hermitian dynamics; and \sec{discussion} concludes and discusses open questions. In the appendices, \cref{APD 1} lists the mathematical prerequisites and definitions used in this paper; \cref{APD 2} provides a specific example of a Schrödinger equation to demonstrate the reasonableness and satisfiability of the main theorem's assumptions; and \cref{apd: supply proof} supplements the detailed proof of the main theorem and the necessary lemmas.

\section{Main Theorem}\label{sec:results}
In this section, we present the main result of this work, which extends the LCHS technique~\cite{ALL23} for Hamiltonian simulation to the case of unbounded operators in infinite-dimensional spaces. Specifically, we focus on simulating unbounded time-evolution Hamiltonians, which represent the most general form of quantum systems with a complex, continuous spectrum.

The primary challenge in proving the simulation formula Inf-LCHS  stems from the unbounded nature of the operators, whose domain is typically restricted to a subset of the Hilbert space. This constraint prevents the direct application of standard simulation techniques. Additionally, the unbounded norm significantly complicates the process of proving the existence of limits or the convergence of series within the theorem's proof. If the integral of the real part of the operators over \( [0, t] \leq T \) is bounded, the limit in \cref{limit} exists, thereby ensuring the existence of the derivative limit, which is the key step in proving the analyticity.

The fundamental difficulty in proving the theorem lies in the need to demonstrate and leverage the fact that the analytic function of the operator-valued function also satisfies the Cauchy Integral Theorem. To achieve this, it is necessary to prove the analyticity of the integrand in \cref{fomula: analytic function}, which means proving the existence of the derivative limit of the function $\mathbf{F}(z) = f(-\mathrm{i} z)e^{\int_0^t(-z L(s)-\mathrm{i} H(s)) \mathrm{d} s}$ with respect to $z$. Since the function $\mathbf{F}(z)$  contains both the real part $L(t)$ and the imaginary part $H(t)$ of the operator $A(t)$, and these components are generally non-commuting, i.e. we have $e^{\int_0^t(-z L(s)-\mathrm{i} H(s)) \mathrm{d} s}\neq e^{-z\int_0^t  L(s)\mathrm{d} s}e^{\int_0^t -\mathrm{i} H(s) \mathrm{d} s}$ in usual.  Therefore, proving the existence of the derivative limit with respect to $z$ for $e^{\int_0^t(-z L(s)-\mathrm{i} H(s)) \mathrm{d} s}$ is a non-trivial task. 

In this paper, we denote $X$ as a Hilbert space and $L(X)$ as the set of linear operators on $X$. As usual, we assume that time is a continuous variable, and thus require continuity of the strong operator topology in $\mathcal{D}(A) = \{\phi \in X \mid \phi \in \mathcal{D}(A(t))\}$: $A(t) \in \operatorname{SOT}([0,T], L(X))$. That is, for all $\phi_0 \in \mathcal{D}(A)$, we have the condition $\lim_{t \to t_0} \|A(t)\phi_0 - A(t_0)\phi_0\| = 0$.
\begin{assumption}\label{adjoint_assum}
$A(t)$ is written as the Cartesian decomposition $A(t) = L(t) + \mathrm{i}H(t)$, where $L = \frac{A + A^\dagger}{2}$ and $H = \frac{A - A^\dagger}{2\mathrm{i}}$ are the Hermitian and anti-Hermitian parts of $A(t)$, respectively. There exists a subspace $X_1 \subseteq X$ such that $L$, $H$, the integrals $\int_0^t L(s) \mathrm{d}s$, $\int_0^t H(s)\mathrm{d}s$, and $\int_0^t(kL(s) + H(s)) \mathrm{d}s$ for all $k \in \mathbb{R}$ are self-adjoint on it. Additionally, we require that $X(t) + Y(t)$ is densely defined in $X$, ensuring that $X(t) + zY(t)$ is densely defined for all $z \in \mathbb{C}$, as $\mathcal{D}(X(t) + zY(t)) = \mathcal{D}(X) \cap \mathcal{D}(Y)$.
\end{assumption}

\begin{assumption}\label{assum: func}
   Function $f(z):\mathbb{C}\to\mathbb{C}$ is analytic, decaying, and normalized, which means it satisfies the following three properties representatively:
\begin{enumerate}
        \item (Analyticity) $f(z)$ is analytic on the lower half plane $\{z: \operatorname{Im}(z)<0\}$ and continuous on $\{z: \operatorname{Im}(z) \leq 0\}$,
        \item (Decay) there exists a parameter $\alpha>0$ such that $|z|^\alpha|f(z)| \leq \widetilde{C}$ for a constant $\widetilde{C}$ when $\operatorname{Im}(z) \leq 0$, and
        \item (Normalization) $\int_{\mathbb{R}} \frac{f(k)}{1-\mathrm{i} k} \mathrm{~d} k=1$.
    \end{enumerate} 
\end{assumption}

\begin{assumption}\label{assum: domain}
We choose a subspace $X_0$ of $X$ as:
        $$X_0 = \left\{ \psi_0 \in X \mid \mathcal{T} e^{-\mathrm{i} \int_0^t (k L(s) + H(s)) \, \mathrm{d}s} \psi_0 \in \mathcal{D}(L(t) + H(t)), \, \forall (k,t) \in \mathbb{R} \times [0, T] \right\} \cap X_1,$$
    where $X_1$ is the space that guarantees the self-adjointness of the operators defined in \cref{adjoint_assum}.
\end{assumption}

In \cref{assum: domain}, the definition of $X_0$ ensures that the unitary time-ordering operator $\mathcal{T} e^{-\mathrm{i} \int_0^t (k L(s) + H(s)) \, \mathrm{d}s}$ acting on the initial vector $\psi_0$ is in the domain of the operator $-\mathrm{i} k L(t) - \mathrm{i} H(t)$, which appears on the right-hand side of the first equation in \cref{main_thm_formula1}. In fact, this choice is made solely to ensure mathematical rigor during the proof of the theorem (specifically, in \cref{main_thm_formula1}). However, in the final expression of the Inf-LCHS  formula, which only involves unitary operators, the domain issue does not need to be considered in numerical simulations. Therefore, the definition of $X_0$ could potentially be broader.

\begin{theorem}[Inf-LCHS Theorem]\label{main_thm}
     Suppose $A(t) = L(t) + \mathrm{i}H(t)$, $L(t)$ and $H(t)$ satisfy \cref{adjoint_assum}. We define $Y(t)=\int_0^t L(s)\mathrm{~d}s$ and $X(t)=-\mathrm{i}\int_0^t H(s)\mathrm{~d}s$. Suppose that $f(z)$ is a function of $z \in \mathbb{C}$ that satisfies \cref{assum: func}. If operators $L(t)$ and $H(t)$ satisfy one of the following two conditions:
     \begin{enumerate}[label=(\alph*)]
         \item \label{condition1}operator $Z(X(t), wY(t)-X(t))$ exists for all $w$ in the lower half plane $\{z: \operatorname{Im}(z)<0\}$,
         \item \label{condition2}$Y(t)$ is a bounded operator,
     \end{enumerate} 
 and $L(s) \succeq \lambda_0>0$ for all $0 \leq s \leq t$, then $\forall \psi_0\in X_0$ defined in \cref{assum: domain} ($\mathcal{T}$ denotes the time-ordering operator)
\begin{equation}\label{main_thm_equ}
    \psi_t =\int_{\mathbb{R}} \frac{f(k)}{1-\mathrm{i} k} \mathcal{T} e^{-\mathrm{i} \int_0^t(k L(s)+H(s)) \mathrm{d} s} \mathrm{~d} k\psi_0
\end{equation}
is the solution of the differential equation $\frac{\partial u(x,t)}{\partial t}=-A(x,t) u(x,t)$ with the initial vector $u(x,0)=\psi_0$.

\end{theorem}

\begin{proof}
We will show that $\psi_t$ satisfies the PDE $\frac{\partial u(x,t)}{\partial t}=-A(x,t) u(x,t)$ with the initial vector $u(x,0)=\psi_0$.
First, for the initial condition, by plugging in $t=0$ and using the normalization condition in \cref{assum: func}, we have
\begin{equation}\label{RHS: initial state}
    \psi_t|_{t=0}=\int_{\mathbb{R}} \frac{f(k)}{1-\mathrm{i} k} \mathrm I \mathrm{~d} k \psi_0 = \psi_0
\end{equation}
Now, we choose a fixed $\delta>0$ and consider $t \in[\delta, T]$. Differentiating $\psi_t$ with respect to $t$ yields

\begin{align}\label{main_thm_formula1}
\frac{\mathrm{d} \psi_t}{\mathrm{d} t} \notag=&\mathcal{P} \int_{\mathbb{R}} \frac{f(k)}{1-\mathrm{i} k}  (-\mathrm{i} k L(t)-\mathrm{i} H(t))\mathcal{T} e^{-\mathrm{i} \int_0^t(k L(s)+H(s)) \mathrm{d} s} \mathrm{~d} k\psi_0 \notag \\
 =&\mathcal{P} \int_{\mathbb{R}} \frac{f(k)}{1-\mathrm{i} k}\left( \left[\left(-L(t)-\mathrm{i} H(t)\right)+\left(1-\mathrm{i} k\right) L(t)\right] \mathcal{T} e^{-\mathrm{i} \int_0^t(k L(s)+H(s)) \mathrm{d} s}\right) \mathrm{~d} k \psi_0\notag\\
 =&\mathcal{P} \int_{\mathbb{R}} \frac{f(k)}{1-\mathrm{i} k}\left[-A(t)+\left(1-\mathrm{i} k\right) L(t)\right] \mathcal{T} e^{-\mathrm{i} \int_0^t(k L(s)+H(s)) \mathrm{d} s} \mathrm{~d} k \psi_0\notag\\
 =&-A(t) \psi_t+L(t)\left(\mathcal{P} \int_{\mathbb{R}} f(k) \mathcal{T} e^{-\mathrm{i} \int_0^t(k L(s)+H(s)) \mathrm{d} s} \mathrm{~d} k \psi_0\right)
\end{align}
We consider the Cauchy principal value ($\mathcal{P}$) of the integrals, as absolute convergence may not be guaranteed. Additionally, in the first equation, we interchange the differentiation and the integral. The validity of this operation is demonstrated in \cref{APD: order change}. We also exchange the limit (i.e., the Cauchy principal value) and the differentiation, which is rigorous by the uniform convergence of the right-hand side in the first equation over $t \in [\delta, T]$. Moreover, the vector inside the integral in \cref{main_thm_formula1} is well-defined according to \cref{assum: domain}.

According to \cref{lemma_main}, the second term on the last line of \cref{main_thm_formula1} vanishes. So we obtain
\begin{align}\label{main_thm_ode}
    \frac{\mathrm{d} \psi_t}{\mathrm{d} t}=-A(t) \psi_t
\end{align}
for all $t \in[\delta, T]$. As $\delta$ can be an arbitrary positive number, \cref{main_thm_ode} holds for all $t \in(0, T]$. Together with \cref{RHS: initial state}, we prove that $\psi_t$ satisfies the PDE $\frac{\partial u(x,t)}{\partial t}=-A(x,t) u(x,t)$ with the initial vector $u(x,0)=\psi_0$ and complete the proof.
\end{proof}

We have assumed \( L(t) \succeq \lambda_0 > 0 \) for a positive number \( \lambda_0 \) and established the Inf-LCHS theorem. To extend the theorem to the case where \( L(t) \succeq 0 \), we need to identify a fixed, non-degenerate common subset of all \( X_0(\lambda) \) in which the initial vector \( \psi_0 \) remains as \( \lambda \to 0 \), where \( X_0(\lambda_0) \) is the space defined in \cref{assum: domain} under the assumption \( L(t) \succeq \lambda_0 > 0 \). Additionally, we must show that the solutions of the PDE depend continuously on \( L(t) \). If these conditions are satisfied, it may be possible to establish the theorem for \( L(t) \succeq 0 \) as well.

The proof of \cref{main_thm} presented here is a high-level demonstration that serves as the final step, following the completion of all conditions and formula derivations. A more detailed, step-by-step proof is provided in \cref{sec:proof} and \cref{apd: supply proof}.

\section{Main Theorem Conditions: Derivation and Proof}\label{sec:proof}
In this section, we identify the conditions required to address the key difficulty in the proof (the unboundedness of the operators) and demonstrate how the Inf-LCHS  formula holds under these conditions. Only two key lemmas that describe these conditions are presented in the main text, while the remaining auxiliary proofs are provided in \cref{apd: supply proof}.

To overcome the difficulties arising from the unbounded norms, we employ two techniques. One approach in \cref{lemma_limit} is to ensure the existence of the operator $Z(X(t), zY(t)-X(t))$, which allows the Baker-Campbell-Hausdorff (BCH) formula to hold, thereby establishing the existence of the derivative. However, in practical applications, the operator $Z(X(t), zY(t)-X(t))$ defined in \cref{def: Z function} is an operator-valued series with unbounded terms, making the verification of series convergence particularly challenging. In such cases, we typically demonstrate the convergence by verifying that the nested Lie brackets of the real and imaginary parts of the Hamiltonians of specific quantum systems vanish at a given level, thereby reducing the infinite summation of series terms to a finite sum. Another method in \cref{lemma:limit_trotter} is to apply the Trotter decomposition formula for unbounded operators that satisfy the conditions in \cref{lemma:Trotter}, which makes the derivative operation valid and guarantees the existence of the derivative.

In the following two lemmas, we give two conditions that guarantee the  analyticity of  $f(-\mathrm{i} z)e^{\int_0^t(-z L(s)-\mathrm{i} H(s)) \mathrm{d} s}$, and Cauchy's integral theorem holds by \cref{Cauchy's integral theorem}. Both lemmas in this section are discussed under the assumption that $X$ is a Hilbert space, $L(X)$  denotes the set of linear operators on $X$, and $A(\cdot) \in  \operatorname{SOT}([0,T],L(X))$. Let $A(t) = L(t) + \mathrm{i}H(t)$, where $L(t)$ and $H(t)$ satisfy \cref{adjoint_assum}. 

In this paper, we define the nested Lie brackets expansion of the right hand side in \cref{BCH} as the notation:
\begin{align}\label{def: Z function}
    Z(X, Y):\notag
    = & \left.X+Y+\frac{1}{2}[X, Y]+\frac{1}{12}([X,[X, Y]]]+[Y,[Y, X]]\right.) -\frac{1}{24}[Y,[X,[X, Y]]] \\\notag
& -\frac{1}{720}([Y,[Y,[Y,[Y, X]]]]+[X,[X,[X,[X, Y]]]]) \\\notag
& +\frac{1}{360}([X,[Y,[Y,[Y, X]]]]+[Y,[X,[X,[X, Y]]]]) \\
& +\frac{1}{120}([Y,[X,[Y,[X, Y]]]]+[X,[Y,[X,[Y, X]]]]) +\cdots
\end{align}
\begin{lemma}\label{lemma_limit}
Let  $A(t) = L(t) + \mathrm{i}H(t)$, $Y(t)=\int_0^t L(s)\mathrm{~d}s$ and $X(t)=-\mathrm{i}\int_0^t H(s)\mathrm{~d}s$. Suppose that $-z Y(s)+ X(s)$ is densely defined and an operator $Z(X(t), zY(t)-X(t))$ exists for $z\in \mathcal{D}\subseteq\mathbb{C}$, then $\mathbf{F}(z) = f(-\mathrm{i} z)e^{\int_0^t(-z L(s)-\mathrm{i} H(s)) \mathrm{d} s}$ is analytic in $\mathcal{D}$ for any numerical analytic function $f(z)$.
\end{lemma}
\begin{proof}
         Since for arbitrary $z_0\in \mathcal{D}$, we have
    \begin{align}\label{equ: exist of deriv}
    \notag & e^{X(t)} \left.
       \dv{z}\right|_{z=z_0}e^{\int_0^t(z L(s)+\mathrm{i} H(s)) \mathrm{d} s} \\\notag
    = &\left.\dv{z}\right|_{z=z_0} e^{X(t)}e^{zY(t)-X(t)}\\\notag
    = & \left.\dv{z}\right|_{z=z_0} e^{Z(X(t), zY(t)-X(t))}\\\notag
    = & \left.\dv{z}\right|_{z=z_0} e^{z(Y(t) + \frac{1}{2}[X(t), Y(t)]+\frac{1}{12}([X(t),[X(t), Y(t)]]]+z[Y(t),[Y(t), X(t)]]) + \cdots)}\\
    = & (Y(t) + \frac{1}{2}[X(t), Y(t)]+\frac{1}{12}([X(t),[X(t), Y(t)]]]+z_0[Y(t),[Y(t), X(t)]]) + \cdots)e^{Z(X(t), z_0Y(t)-X(t))},
    \end{align}
    where the second equation is due to the existence of $Z(X(t), zY(t)-X(t))$ and the BCH formula \cref{BCH}, the third equation holds by the expansion of $Z(X(t), zY(t)-X(t))$. \cref{equ: exist of deriv} demonstrates the existence of the derivative of the operator-valued function $e^{\int_0^t(z L(s)+\mathrm{i} H(s)) \mathrm{d} s}$. Consequently, $e^{\int_0^t(z L(s)+\mathrm{i} H(s)) \mathrm{d} s}$ is analytic, and therefore $\mathbf{F}(z) = f(-\mathrm i z)e^{\int_0^t(-z L(s)-\mathrm{i} H(s)) \mathrm{d} s}$ is also analytic.
\end{proof}

\begin{lemma}\label{lemma:limit_trotter}
     Let  $A(t) = L(t) + \mathrm{i}H(t)$, $Y(t)=\int_0^t L(s)\mathrm{~d}s$ and $X(t)=-\mathrm{i}\int_0^t H(s)\mathrm{~d}s$. Denote $\mathcal{D}=\{z|\Re(z)\geq0, z\in\mathbb{C}\}$. If $Y(t)$ is bounded and $L(t) \succeq \lambda_0>0$ for a positive number $\lambda_0$, and operator $X(t)+ zY(t)$ is densely defined for all $z\in \mathbb{C}$, then $\mathbf{F}(z) = f(-\mathrm{i} z)e^{\int_0^t(-z L(s)-\mathrm{i} H(s)) \mathrm{d} s}$ is analytic in $\mathcal{D}$ for any numerical analytic function $f(z)$.
\end{lemma}
\begin{proof}
    We first prove the Trotter decomposition of unbounded operators in \cref{lemma:Trotter} holds under this condition. Let $A=-(\mathrm{i}X(t)+\Im(z)Y(t)),\ B=-\Re(z)Y(t)$, then $A+\mathrm{i}B=-\mathrm{i}(X(t)+\bar zY(t))$ is densely defined, and thus $-\mathrm{i}(\mathrm{i}X(t)+\Im(z)Y(t))-\Re(z)Y(t) = - (X(t)+\bar zY(t))^\dagger$ is closed. It suffices to check $-\mathrm{i}(\mathrm{i}X(t)+\Im(z)Y(t))$ and $-\Re(z)Y(t)$ generate the contraction semigroups $e^{\mathrm{i}\lambda (\mathrm{i}X(t)+\Im(z)Y(t))}$ and $e^{-\lambda \Re(z)Y(t)}$. By the Hille–Yosida theorem, since $-\mathrm{i}(\mathrm{i}X(t)+\Im(z)Y(t))-\Re(z)Y(t)=X(t) - zY(t)$ is closed and densely defined and $-\Re(z)Y(t)$ is a bounded operator, we only need to prove the fact that every real $\lambda>0$ belongs to the resolvent set of $U$ and for such $\lambda$, $\left\|(\lambda I-U)^{-1}\right\| \leq \frac{1}{\lambda}$
holds for $U=-\mathrm{i}(\mathrm{i}X(t)+\Im(z)Y(t))-\Re(z)Y(t)$ and $U=-\Re(z)Y(t)$.

If $U=-\Re(z)Y(t)$, since $\Re(z)\geq0$ and $L(t) \succeq \lambda_0>0$ for a positive number $\lambda_0$, $2\lambda\Re(z)Y(t)\succeq0$ and $\Re(z)^2Y(t)^2\succeq0$. Therefore,
\begin{equation}\label{contra_1}
    \left\|\frac{1}{\lambda - (-\Re(z)Y(t))}\right\|^2 = \left\|\frac{1}{(\lambda +\Re(z)Y(t))^2}\right\| =\left \|\frac{1}{\lambda^2 +2\lambda\Re(z)Y(t) +\Re(z)^2Y(t)^2 }\right\|\leq\frac{1}{\lambda^2}.
\end{equation}
And for the case when $U=-\mathrm{i}(\mathrm{i}X(t)+\Im(z)Y(t))-\Re(z)Y(t)$, since $[\Re(z)Y(t) + \mathrm{i}(\mathrm{i}X(t)+\Im(z)Y(t))][(\Re(z)Y(t) - \mathrm{i}(\mathrm{i}X(t)+\Im(z)Y(t)))] = [\Re(z)Y(t) + \mathrm{i}(\mathrm{i}X(t)+\Im(z)Y(t))]\cdot[\Re(z)Y(t) + \mathrm{i}(\mathrm{i}X(t)+\Im(z)Y(t))]^\dagger\succeq0$, the similar calculation gives
\begin{align}\label{contra_2}
    \notag
    &\left\|\frac{1}{\lambda + \Re(z)Y(t) + \mathrm{i}(\mathrm{i}X(t)+\Im(z)Y(t))}\right\|^2 \\  \notag
    = &\left\|\frac{1}{[\lambda + \Re(z)Y(t) + \mathrm{i}(\mathrm{i}X(t)+\Im(z)Y(t))]\cdot [\lambda + \Re(z)Y(t) + \mathrm{i}(\mathrm{i}X(t)+\Im(z)Y(t))]^\dagger}\right\| \\ \notag
    =& \left\|\frac{1}{[\lambda + \Re(z)Y(t) + \mathrm{i}(\mathrm{i}X(t)+\Im(z)Y(t))]\cdot [\lambda + \Re(z)Y(t) - \mathrm{i}(\mathrm{i}X(t)+\Im(z)Y(t))]}\right\|\\ \notag
    =&\left \|\frac{1}{\lambda^2 +2\lambda\Re(z)Y(t) + [\Re(z)Y(t) + \mathrm{i}(\mathrm{i}X(t)+\Im(z)Y(t))][(\Re(z)Y(t) - \mathrm{i}(\mathrm{i}X(t)+\Im(z)Y(t)))]}\right\|\\
    \leq&\frac{1}{\lambda^2}.
\end{align}
Therefore, all the requirements in \cref{lemma:Trotter} are satisfied, thus the Trotter decomposition of $e^{\int_0^t-(z L(s)+\mathrm{i} H(s)) \mathrm{d} s}$ holds. Thus, for all $\psi$ in the Hilbert space $X$,  
    \begin{align}\label{limit}
       \left.\dv{z}\right|_{z=z_0}e^{\int_0^t-(z L(s)+\mathrm{i} H(s)) \mathrm{d} s} \psi\notag
    = &\left.\dv{z}\right |_{z=z_0}\lim_{n\to\infty}\left(\left(e^{-\Re(z)Y(t)/n}e^{- \mathrm{i}(\mathrm{i}X(t)+\Im(z)Y(t))/n}\right)^n\psi\right)\\\notag
    = & \left.\lim_{n\to\infty}\left(\dv{z}\right |_{z=z_0}\left(e^{-\Re(z)Y(t)/n}e^{- \mathrm{i}(\mathrm{i}X(t)+\Im(z)Y(t))/n}\right)^n\psi\right)\\
    = & \lim_{n\to\infty}\left(\left(-\frac{Y(t)}{n} e^{-\Re(z_0)Y(t)/n}e^{- \mathrm{i}(\mathrm{i}X(t)+\Im(z_0)Y(t))/n}\right)^n\psi\right).
    \end{align}
    Since $Y(t)$ is bounded, thus the limit in \cref{limit} exists, which means the derivative of $e^{\int_0^t(-z L(s)-\mathrm{i} H(s)) \mathrm{d} s}$ exists for all $z\in\mathcal{D} $. Therefore, $\mathbf{F}(z) = f(-\mathrm{i} z)e^{\int_0^t(-z L(s)-\mathrm{i} H(s)) \mathrm{d} s}$ is analytic in $\mathcal{D}$ for any numerical analytic function $f(z)$.
\end{proof}

The lemmas presented above specify the additional conditions required to extend the finite-dimensional LCHS formula to infinite-dimensional unbounded operators. Specifically, to address the challenges posed by the unbounded nature of the Hamiltonian \( A(t) \) in continuous spectrum quantum systems, we consider one of two approaches.

The first approach assumes the existence of the operator \( Z(X(t), wY(t) - X(t)) \) as presented in \cref{lemma_limit}, which is expressed as an infinite sum of nested Lie brackets. This assumption fundamentally guarantees the existence of the summation limit, thus replacing the bounding requirement of the first approach. However, verifying the existence of such an operator, represented as an infinite series of nested Lie brackets, can be computationally challenging in practice. To simplify this verification, we often demonstrate that all nested Lie brackets of \( L(t) \) and \( H(t) \) of a certain order vanish, such as
$$
[L, [L,[\dots, [L, H]]]] = [L, [H,[\dots, [L, H]]]] = \dots = [H, [L,[\dots, [L, H]]]] = [H,[H,[ \dots, [L, H]]]] = 0.
$$
In this case, since each nested Lie bracket of \( X(t) \) and \( wY(t) - X(t) \) in \( Z(X(t), wY(t) - X(t)) \) is equivalent to the corresponding nested Lie bracket of \( X(t) \) and \( wY(t) \) (because the Lie bracket of \( X(t) \) with itself is zero and the Lie bracket satisfies bilinearity), all higher-order nested Lie brackets vanish. Consequently,
$$
Z(X(t), wY(t) - X(t))
$$
reduces to a finite sum of operators, thereby ensuring its existence.
This condition can be interpreted as a form of commutativity between the real and imaginary parts of the operator in a specific sense. This is a surprising observation because the form of Inf-LCHS theorem \cref{main_thm_equ} itself does not explicitly represent any commutativity between the real and imaginary parts. The formula merely involves integrating the sum of the real and imaginary parts in the exponential, but the formula holds when such a commutative relationship exists between them.

The second approach, outlined in \cref{lemma:limit_trotter}, assumes that the integral of the real part of operator $A(t)$ over \( [0, t] \leq T \), denoted as \( Y(t) \), is a bounded operator. This enables the application of the Trotter decomposition to establish the existence of the derivative limit. In practical scenarios, such as the Schrödinger equation with a complex potential \( V(x, t) = V_R(x, t) + \mathrm{i} V_I(x, t) \), the imaginary part \( V_I(x, t) \) of the potential often has compact support.  Similarly, the integral $\int_0^t V_I(x, t) $ is bounded, which makes it a bounded operator when viewed as the integral of the real part \( L(t) \) of the Hamiltonian governing the quantum system described by the equation.

Either of these two conditions constitutes the sole additional requirement for the applicability of the Inf-LCHS theorem to quantum systems with complex, continuous spectra. As the most general and comprehensive extension of the LCHS formula, the Inf-LCHS theorem provides a robust framework for such systems. These two conditions are referred to as conditions \ref{condition1} and \ref{condition2} in \cref{main_thm}.

\section{Inf-LCHS Implementation and Sampling Methods}\label{sec:sample}
In this section, we first discuss the selection of the kernel function $f$ and the numerical discretization required for implementing the algorithm on quantum computers.

A potential choice for the kernel function $f$ is given by 
\begin{equation}
    f(z) = \frac{1}{\pi(1+\mathrm{i}z)}.
\end{equation}
However, other kernel functions with improved convergence properties exist. For instance, as described in~\cite{ACL23}, the function
\begin{equation}
    f(z) = \frac{1}{2\pi \exp(-2^\beta) \exp((1+\mathrm{i}z)^\beta)}, \quad (0<\beta<1)
\end{equation}
is considered optimal in the sense that the truncation range of the infinite integral cannot be significantly reduced beyond this choice.

To address the truncation of the infinite integral, we introduce the simplified notations:
\begin{equation}
    g(k):=\frac{f(k)}{1-\mathrm{i}k}, \quad \mathcal{U}(k,T) := \mathcal{T} \exp\left(-\mathrm{i}\int_0^T \big(kL(s) + H(s)\big)\,\mathrm{d}s\right),
\end{equation} 
where $\mathcal{T}$ denotes the time-ordering operator. To bound the truncation error, we impose the condition
\begin{equation}
    \left\lVert \int_{\mathbb{R}} g(k) \mathcal{U}(k,T)\,\mathrm{d}k  - \int_{-K}^K g(k) \mathcal{U}(k,T)\,\mathrm{d}k \right\rVert < \varepsilon.
\end{equation}

Since $\mathcal{U}(k,T)$ is unitary, for any $k$ and $T$, we obtain the bound
\begin{equation}
    \left\lVert \int_{\mathbb{R}} g(k) \mathcal{U}(k,T)\,\mathrm{d}k  - \int_{-K}^K g(k) \mathcal{U}(k,T)\,\mathrm{d}k \right\rVert 
    \le \int_{-\infty}^{-K} \lvert g(k)\rvert\,\mathrm{d}k  + \int_{K}^{+\infty} \lvert g(k)\rvert\,\mathrm{d}k.
\end{equation}

Using the rapid decay property of $g(k)$ stated in Lemma 9 of~\cite{ACL23}, it suffices to choose the truncation parameter $K$ as
\begin{equation}\label{eq:complexity_K}
    K = \mathcal{O}\left(\log(1/\varepsilon)^{1/\beta}\right).
\end{equation}

Next, we estimate the integral using numerical quadrature schemes and bound the quadrature error.
Two different schemes are proposed for this purpose.

 In the first scheme, we assume $L(t)$ is already approximated by a finite-dimensional matrix or a bounded linear operator. This is a scenario often considered in numerical simulations. For example, if $L(t)$ represents the Laplace operator, the spatial variable can be discretized using a finite number of grid points, and the differential operator can be approximated by a matrix via finite difference methods. Under this assumption, the norm of $L(t)$ becomes well-defined, enabling us to analyze the precision based on the approach in~\cite{ACL23}, which achieves state preparation costs with near-optimal dependence on all relevant parameters.
  
In the second scheme, we employ Monte Carlo integration. This method yields an error bound independent of the derivatives of $A(t)$, thereby avoiding the complications associated with unbounded operators.

The quantum implementation of the LCHS formula relies on the linear combination of unitaries (LCU) technique \cite{Childs2012HamiltonianSU,GY2008}. In both schemes discussed above, we obtain the approximation
$
u(T) \approx \sum_{j=0}^{N-1} C_j \mathcal{U}(k_j, T) u_0,
$ using the Inf-LCHS integral transform formula together with the discretization of the integrals. 
Here $u(T)$ represents the solution at time $T$ and $u_0$ is the initial state. In brief, the LCU cost is proportional to the total number $N$ of unitaries and to the cost of each Hamiltonian simulation $\mathcal{U}(k_j, T)$. Thus, we need to analyze the scaling of $N$, which corresponds to the total number of quadrature nodes in a numerical integration scheme or the total number of samples in a Monte Carlo integration scheme. From this, we would be able to estimate the query complexity for $A(t)$ (more precisely, the HAM-T oracle for $kL+ H$) and $u_0$ (more precisely, the state preparation oracle for $u_0$), based on the results in~\cite{ALL23}.

\subsection{Inf-LCHS-Gaussian: Composite Gaussian Quadrature}
We consider the composite Gaussian quadrature scheme, defined as
\begin{equation}\label{eq:quadrature}
    \norm{\int_{-K}^K g(k) \mathcal{U}(k,T)\dd k - \sum_{m=-M}^{M-1}\sum_{q=0}^{Q-1}c_{q,m}\mathcal{U}(k_{q,m},T)} < \varepsilon.
\end{equation}
The integration interval $[-K, K]$ is divided into $2M$ subintervals, each of length $h = K/M$. On each subinterval $[mh, (m+1)h]$ for $m = -M, \dots, M-1$, we apply Gaussian quadrature with $Q$ nodes. Denoting the nodes and weights by $k_{q,m}$ and $w_{q,m}$ respectively, we define $c_{q,m} := w_{q,m} \cdot g(k_{q,m})$.

\begin{theorem}[Inf-LCHS-Gaussian]\label{THM:Inf-LCHS-G}
Let $\sum_{m=-M}^{M-1}\sum_{q=0}^{Q-1}c_{q,m}\mathcal{U}(k_{q,m},T)$ denote the composite Gaussian quadrature of the integral $\int_{-K}^{K} g(k) \mathcal{U}(k,T)\dd k$. To ensure the quadrature error is bounded by $\varepsilon$ (\cref{eq:quadrature}), the total number of nodes scales with 
\[
\mathcal{O}\left( T \max_{0 \le t \le T} \| L(t) \| \left(\log\frac{1}{\varepsilon}\right)^{1+1/\beta}\right).
\]
\end{theorem}

\begin{proof}
According to Lemma 10 in~\cite{ACL23}, the local error for each subinterval satisfies
\begin{equation}
    \norm{\int_{mh}^{(m+1)h} g(k)\mathcal{U}(k,T)\dd k - \sum_{q=0}^{Q-1}c_{q,m}\mathcal{U}(k_{q,m},T)} \le h^{2Q-1} A(\beta) \left(\frac{eT \max_{0 \le t \le T} \| L(t) \|}{2}\right)^{2Q} \le \frac{1}{2^{2Q}} A(\beta),
\end{equation}
where $A(\beta)$ is a constant depending only on $\beta$. By setting the step size as $h = \frac{1}{eT \max_{0 \le t \le T} \| L(t) \|}$, one can ensure that the global quadrature error is bounded by $\mathcal{O}(K \cdot 2^{-2Q})$. Thus, choosing
\begin{equation}
    Q = \mathcal{O}\left(\log\frac{K}{\varepsilon}\right) = \mathcal{O}\left(\log\frac{1}{\varepsilon}\right),
\end{equation}
and noting $K = \mathcal{O}(\log(1/\varepsilon)^{1/\beta})$ (\cref{eq:complexity_K}), yields the total number of nodes:
\begin{equation}
    N = 2M Q = \frac{2KQ}{h} = \mathcal{O}\left(T \max_{0 \le t \le T} \| L(t) \| \left(\log\frac{1}{\varepsilon}\right)^{1+1/\beta}\right).
\end{equation}
\end{proof}

 Once the circuit for $\mathcal{U}(k,T)$ is constructed, the summation $\sum_{m,q} c_{q,m} \mathcal{U}(k_{q,m}, T)$ can be implemented using the LCU technique. As shown in~\cite{LW18, ACL23}, the query complexity of $A(t)$ of a single run of the algorithm is $\widetilde{\mc O}(\alpha_A K T\log(1/\varepsilon)
 )$, 
where $\alpha_A \geq \| A(t) \|$ is the subnormalization factor. Due to the success probability of choosing the correct subspace in the post-selection step of LCU, the overall complexity should be multiplied with $\frac{\norm{u_0}}{\norm{u(T)}}$. Therefore, we have the following query complexity result:
\begin{corollary}
    In the Inf-LCHS-Gaussian implementation, we can solve the linear differential equation at time $T$ with error $\varepsilon$ using $\widetilde{\mathcal{O}}\left(\frac{\| u_0 \|}{\| u(T) \|} \alpha_A T \left(\log\frac{1}{\varepsilon}\right)^{1+1/\beta}\right)$ queries to the matrix input oracle for $A(t)$ and $\mathcal{O}\left(\frac{\| u_0 \|}{\| u(T) \|}\right)$ queries to the state preparation oracle of the initial state $\ket{u_0}$, and by choosing the total number of unitaries to be $N = \mathcal{O}\left(\alpha_A T  \left(\log\frac{\norm{u_0}}{\norm{u(T)}\varepsilon}\right)^{1+1/\beta}\right)$
\end{corollary}

\subsection{Inf-LCHS-MC: Monte Carlo Integration}
We adopt a Monte Carlo quadrature scheme, where
\begin{equation}
    \int_{-K}^K g(k) \mathcal{U}(k,T) \dd k \approx \frac{2K}{N_s}\sum_{i=1}^{N_s} g(\xi_i) \mathcal{U}(\xi_i, T) = \frac{2K}{N_s} \sum_{i=1}^{N_s} F(\xi_i, T) = 2K \cdot \langle F \rangle_{N_s}.
\end{equation}
Here we denote
\begin{equation}
    F(k,T):= g(k) \mc U(k,T)
\end{equation} for simplicity. The nodes $\{\xi_i\}_{i=1}^{N_s}$ are sampled uniformly from $[-K, K]$, and $N_s$ is the number of samples.

\begin{theorem}[Inf-LCHS-MC]\label{THM:Inf-LCHS-MC}
Let $\frac{2K}{N_s}\sum_{i=1}^{N_s} g(\xi_i) \mathcal{U}(\xi_i, T)$ denote the Monte Carlo approximation of the integral $\int_{-K}^K g(k) \mathcal{U}(k,T) \dd k$. To bound the error by $\varepsilon$, the total number of samples scales with 
$\mathcal{O}\left(\frac{1}{\varepsilon^2} \left(\log\frac{1}{\varepsilon}\right)^{2/\beta}\right).$
\end{theorem}

\begin{proof}
Let $I_{N_s} := \frac{2K}{N_s} \sum_{i=1}^{N_s} g(\xi_i) \mathcal{U}(\xi_i, T)$. We denote $F=g(X) \mc U(X,T)$ with $X\sim \mathrm{Uniform}([-K,K])$, then the variance of $I_{N_s}$ is 
\begin{equation}
    \mathrm{Var}(I_{N_s}) = \frac{4K^2}{N_s} \mathrm{Var}(F).
\end{equation}
Since $\|F(k)\| = |g(k)| \leq \frac{1}{\sqrt{1+k^2}} \leq 1$ and $\norm{F(k)^2} = \abs{g(k)}^2 \le \frac1{1+k^2}\le 1$, it follows that $\mathrm{Var}(F)  = \mathbb E (F^2) - (\mathbb E(F))^2\leq 1$. Therefore, the standard error of the mean is bounded by

\begin{equation}
   \sqrt{\mathrm{Var}(I_{N_s})} \leq \frac{2K}{\sqrt{N_s}}.
\end{equation}
Thus, to achieve error $\varepsilon$, we require $N_s = \mathcal{O}\left(\frac{K^2}{\varepsilon^2}\right)$. Since $K = \mathcal{O}\left(\log(1/\varepsilon)^{1/\beta}\right)$, the total complexity is 
\begin{equation}
    N = N_s = \mathcal{O}\left(\frac{1}{\varepsilon^2} \left(\log\frac{1}{\varepsilon}\right)^{2/\beta}\right).
\end{equation}
\end{proof}

From the view of the number of unitaries in the Inf-LCHS theorem, the Monte Carlo integration has a higher asymptotic complexity than Gaussian quadrature. However, it entirely circumvents the handling of unbounded operators. Even nicer, the scheme is inherently independent of the evolution time $T$ thanks to the fact that each    
$\mathcal{U}(\xi, T)$ is unitary. 
This property ensures robust scaling with respect to long evolution durations, making it especially useful in cases where $T$ is large, $A(t)$ is an unbounded operator or poorly bounded, or the norm $\norm{L(t)}$ is either inaccessible or hard to estimate. Some more explicit examples of the algorithmic applications regarding PDEs are given in \sec{applications}.

\section{Applications}\label{sec:applications}

We apply our Inf-LCHS theorem \cref{main_thm} to simulate non-Hermitian dynamics, including linear parabolic PDEs, queueing models (birth-or-death processes), Schrödinger equations with complex potentials, Lindblad equations, black hole thermal field equations, etc. We discuss the number of unitaries needed in the sampling methods for different dynamics and problem settings.

\subsection{Linear Parabolic PDEs}\label{apply: para}

Linear parabolic partial differential equations (PDEs) are broadly used to describe various evolution behaviors of physical phenomena, such as mass, energy, fluids, particles, etc. Typical examples include the transport equation, heat equation, advection-diffusion equation, Black-Scholes equation, etc. The parabolic PDE includes unbounded first-order and second-order differential operators, such as the Laplacian operator. 

We consider PDEs with the time variable $t \in [0,T]$ and spatial variable $x = (x_0,\cdots,x_{d-1}) \in [0,1]^d$. The PDE of parabolic type has the general form
\begin{align}
    \frac{\partial u(x,t)}{\partial t} = -\mathcal{L}[u(x,t)]
\end{align}
with
\begin{align}
    \mathcal{L}[u(x,t)] = -\sum_{i,j}\partial_i(a^{ij}(x,t)\partial_ju(x,t)) + \sum_{j}b^j(x,t)\partial_ju(x,t) + c(x,t)u(x,t),
\end{align}
given conditions $u(x,0) = u_0(x), ~x\in\Omega$ and $u(x,t) = 0, ~x\in\partial\Omega$ and $a^{ij}(x,t) = a^{ji}(x,t)$. Here $\partial_j u$ is the gradient of $u$ w.r.t. $x_j$. We call $\mathcal{L}$ the elliptic operator as it satisfies the elliptic condition
\begin{align}
    \sum_{i,j}a^{ij}(x,t) \xi_i\xi_j \ge \theta|\xi|^2, \quad \theta>0.
\end{align}
Steady-state solutions to hyperbolic and parabolic PDEs generally solve elliptic PDEs.

We rewrite the PDE operator in the form $$\mathcal{L} = -\mathrm{i}\left(\mathrm{i}\sum_{j}b^j (x, t)\partial_j) -\sum_{i,j}\partial_i(a^{ij}(x,t)\partial_j\right) + c(x, t).$$ Since $\mathcal{L} = \mathbf L + \mathrm{i}\mathbf H$, we have $\mathbf L= -\sum_{i,j}\partial_i(a^{ij}(x,t)\partial_j) + c(x,t)$, $\mathbf H = -\mathrm{i}\sum_{j}b^j (x, t)\partial_j$.

\begin{proposition}\label{prop: commut partialx}
     Let  $\mathcal{L} = \mathbf L + \mathrm{i}\mathbf H$, $Y(t)=\int_0^t \mathbf L\dd s$ and $X(t)=-\mathrm{i}\int_0^t \mathbf H \dd s$. If $a^{ij}(x,t), \  b^j(x,t)$ and $c(x,t)$ are in $C^2(\Omega\times[0, T])$, and differential operators $\mathcal{L}$ in the linear parabolic partial differential equations satisfy one of the following three conditions:
    \begin{enumerate}
        \item $a^{ij}(x, t)$ and $b^j (x, t)$ are independent with $x$, i.e. $a^{ij}(x, t) = a^{ij}(t)$ and $b^j (x, t) = b^j (t)$, $c(x, t)$ is a polynomial in variable $x$, and the solution is smooth function,
        \item $a^{ij}(x, t)=0$ and  $c(x, t)$ is a polynomial in variable $x$,
        \item $b^j (x, t)=0$,
    \end{enumerate}
 then operator $Z(X(t), wY(t)-X(t))$ exists and \cref{main_thm} holds when other requirements are satisfied.
\end{proposition}
\begin{proof}
 We prove by verifying condition \ref{condition1}. It suffices to prove the operator $Z(X(t), Y(t))$ exists. In this condition, $X(t)=-\mathrm{i}\int_0^t \mathbf H \dd s=-\mathrm{i}\sum_{j}\int_0^tb^j (x, s)\dd s\partial_j$, $Y(t)=\int_0^t \mathbf L\dd s = -\sum_{i,j}\partial_i(\int_0^t a^{ij}(x,s)\dd s\partial_j) + \int_0^tc(x,s)\dd s$ since $a^{ij}$ are in $C^2$. Obviously, if coefficient functions satisfy one of the three conditions in \cref{prop: commut partialx}, the integral of them($\int_0^t a^{ij}(x,s)\dd s,\ \int_0^tb^j (x, s)\dd s$ and $\int_0^tc(x,s)\dd s$) also satisfy that condition. Therefore, if we treat the integral of the coefficient functions as new coefficient functions, while neglecting the constant $-\mathrm{i}$, which does not affect the commutativity, we only need to prove that the operator $Z(L(t), H(t))$ exists without loss of generality. For arbitrary polynomial multiply operator $\mathcal{M}_{p(x)}$, its Lie bracket with self-adjoint operator $\mathrm{i}\frac{\partial}{\partial x}$ equals $\mathcal{M}_{\mathrm{i}\frac{\partial p(x)}{\partial x}}$, since
\begin{align}\label{lemma: first order partial commutant}
      \left[\mathrm{i}\frac{\partial}{\partial x}, \mathcal{M}_{p(x)}\right]u\notag
      = &\mathrm{i}\frac{\partial}{\partial x}\mathcal{M}_{p(x)} u - \mathcal{M}_{p(x)}\mathrm{i}\frac{\partial}{\partial x}u\\\notag
      = &\mathrm{i}\frac{\partial (p(x) u)}{\partial x} - p(x) \mathrm{i}\frac{\partial u}{\partial x}\\\notag
      = &\mathrm{i}\left(\frac{\partial p(x) }{\partial x} u + p(x) \frac{\partial  u}{\partial x} - p(x) \frac{\partial u}{\partial x}\right)\\
      = & \mathcal{M}_{\mathrm{i}\frac{\partial p(x)}{\partial x}}u.
\end{align}
    \begin{enumerate}
    
\item 
For the case when $a^{ij}(x, t)$ and $b^j (x, t)$ are independent with $x$, i.e., $a^{ij}(x, t) = a^{ij}(t)$ and $b^j (x, t) = b^j (t)$, if the solution is smooth function, then $[\sum_{i,j}\partial_i(a^{ij}(t)\partial_j), \mathrm{i}\sum_{j}b^j ( t)\partial_j]=0$ and $Z\left(\sum_{i,j}\partial_i(a^{ij}(t)\partial_j), \mathrm{i}\sum_{j}b^j ( t)\partial_j\right) = \sum_{i,j}\partial_i(a^{ij}(t)\partial_j) + \mathrm{i}\sum_{j}b^j ( t)\partial_j$. Thus  $$
Z\left(\sum_{i,j}\partial_i(a^{ij}(t)\partial_j) + \mathcal{M}_{c(x,t)}, \mathrm{i}\sum_{j}b^j ( t)\partial_j\right) = \sum_{i,j}\partial_i(a^{ij}(t)\partial_j) + \mathrm{i}\sum_{j}b^j ( t)\partial_j + Z\left(\mathcal{M}_{c(x,t)}, \mathrm{i}\sum_{j} b^j(t) \partial_j\right).
$$ 
Since the elements in the nested Lie bracket terms of $Z\left(\mathcal{M}_{c(x,t)}, \mathrm{i}\sum_{j} b^j(t) \partial_j\right)$ are either $\partial_j$ or $\mathcal{M}_{c(x,t)}$, and all the nested Lie bracket terms contain $[\sum_{i,j}\partial_i(a^{ij}(t)\partial_j),  \mathcal{M}_{c(x,t)}]$ which is equal to multiply operator $ \mathrm{i}\sum_{j} b^j(t) \mathcal{M}_{\partial_j c(x,t)}$, as shown in \cref{lemma: first order partial commutant}, the nested Lie brackets will evaluate to zero if the number of operators exceeds the degree of the polynomial $c(x,t)$, or if there are more than one multiplying operator $\mathcal{M}_{c(x,t)}$, since it commutes with $\mathcal{M}_{\partial_j c(x,t)}$. Therefore, $Z(\mathrm{i}\sum_{j} b^j(t) \partial_j, \mathcal{M}_{c(x,t)})$ can be expressed as a finite sum of nested Lie brackets, which guarantees its existence.

\item If $a^{ij}(x, t)=0$, for the same reason in the first case, $Z\left(\mathrm{i}\sum_{j} b^j(x,t) \partial_j, \mathcal{M}_{c(x,t)}\right)$ can be expressed as a finite sum of nested Lie brackets since the Lie brackets between $\mathrm{i}\sum_{j} b^j(x,t) \partial_j$ and $\mathcal{M}_{c(x,t)}$ can be expressed in the form of a multiply operator $\mathrm{i}\sum_{j} b^j(x,t) \mathcal{M}_{\partial_j c(x,t)}$.
\item When $b^j(x, t)=0$, then, $\mathcal{L} = \mathbf L$ is self-adjoint operator, thus $Z(X(t), Y(t))= Y (t)$ exists.
    \end{enumerate}
\end{proof}
By $L(t) \succ0$, $\mathcal{L} = \mathbf L + \mathrm{i}\mathbf H$ and \cref{main_thm}, we can simulate the solution of parabolic PDEs
\begin{align}
    u(x,t) =  \int_{\mathbb{R}} \frac{f(k)}{1-\mathrm{i} k} \mathcal{T} e^{-\mathrm{i} \int_0^t(k \mathbf L(s)+\mathbf H(s)) \mathrm{d} s} \mathrm{~d} k u_0(x)
\end{align}
via the Inf-LCHS theorem.

We examine the number of unitaries required in the Inf-LCHS theorem for simulating linear parabolic partial differential equations (PDEs), utilizing the results presented in \cref{sec:sample}. Specifically, for the Inf-LCHS-Gaussian scheme (\cref{THM:Inf-LCHS-G}), the first step is to discretize the spatial variable in the differential operator. We take the central finite difference scheme as the example. Let us assume that the hypercube $[0,1]^d$ is partitioned into an equidistant grid, with each dimension divided into $N_{\rm grid}$ grid points. Consequently, the grid spacing is given by $h = \frac{1}{N_{\rm grid} - 1}$, and the total number of grid points is $N_{\rm grid}^d$. Denote each grid point as $\boldsymbol{\xi}_k = (\xi_1^k, \xi_2^k, \dots, \xi_d^k)$. According to the central finite difference scheme, the discretized form of the partial differential operator takes the following matrix representation:
\begin{equation}\label{eq:discretize}
    \mathbf L_{\rm grid}  = [\partial_i (a^{ij} \partial_j u) - c]_{\rm grid} = \frac{1}{h^2} \left(a^{ij}_{\boldsymbol{\xi}_k + h \mathbf{e}_i} \left( u_{\boldsymbol{\xi}_k + h \mathbf{e}_i + h \mathbf{e}_j} - u_{\boldsymbol{\xi}_k + h \mathbf{e}_i} \right) - a_{\boldsymbol{\xi}_k}^{ij} \left( u_{\boldsymbol{\xi}_k + h \mathbf{e}_j} - u_{\boldsymbol{\xi}_k} \right) \right) - c_{\boldsymbol{\xi}_k},
\end{equation}
where $a^{ij}_x := a^{ij}(x,t)$, $c_x := c(x,t)$ and $u_x := u(x,t)$ are evaluated at the grid points. From \cref{eq:discretize}, it is evident that the norm of the discretized operator behaves as
\begin{equation}
    \norm{\mathbf L_{\rm grid} } = \mathcal{O}\left(\frac{1}{h^2}\right) = \mathcal{O}(N_{\rm grid}^2), \quad \forall t \in \mathbb{R}^+.
\end{equation}
According to \cref{THM:Inf-LCHS-G}, the total number $N$ of Schr\"odinger equations required to simulate the system scales with
\begin{equation}
    N = \mathcal{O} \left( T \max_{0 \le t \le T} \| \mathbf L_{\rm grid} (t) \| \left( \log \frac{1}{\varepsilon} \right)^{1 + 1/\beta} \right) = \mathcal{O}\left( T N_{\rm grid}^2 \left( \log \frac{1}{\varepsilon} \right)^{1 + 1/\beta} \right).
\end{equation}
We should note that the number $N_{\rm grid}$ depends on the specific discretization method~\cite{AL19,CLLL22, JLY22, JLY23,HJZ23}. 
These different methods may introduce additional errors, the specific manifestations of which depend on the particular form and properties of the partial differential operators.

In contrast, for the Inf-LCHS-MC scheme, as demonstrated in \cref{THM:Inf-LCHS-MC}, the total number $N$ of Schr\"odinger equations scales with
\begin{equation}
    N = \mathcal{O} \left( \frac{1}{\varepsilon^2} \left( \log \frac{1}{\varepsilon} \right)^{2/\beta} \right).
\end{equation}
We observe that, although the Inf-LCHS-MC scheme generally exhibits a worse asymptotic complexity, it is applicable to unbounded operators and avoids the need to estimate the norm of the operator $\mathbf L$. Even when the partial differential operator is discretized in the spatial variable, it typically has a large norm that scales with $N_{\rm grid}^2$ and is highly sensitive to the behavior of the coefficients $a(x,t)$ and $c(x,t)$, which can sometimes be hard to estimate.

\subsection{Queueing Models (Birth-death Processes)}

We consider a general infinitity-dimensional continuous-time Markov chain (CTMC) such as the birth-death processes, with extensive applications in the queueing models. In the queueing theory, we can describe a service system by a state space $\{0,1,2,3,\ldots\}$, where the value corresponds to the number of customers in the system, including any currently in service. The model can be described as a CTMC
\begin{equation}
    \frac{\mathrm{~d} \pi}{\mathrm{~d} t} = \pi Q,
\end{equation}
with a transition rate matrix $Q$, whose each row sums to zero.

For example, the M/M/1 model represents a single-server queue where arrivals occur at rate $\lambda$ according to a Poisson process and move the process from state $i$ to state $i+1$, and service times have an exponential distribution with the mean service time $1/\mu$. The transition rate matrix $Q$ is
\begin{equation}
    \left( \begin{array}{cccccccccc}
        -(\lambda+\mu) & \lambda & & & & \\
        \mu & -(\lambda+\mu) & \lambda & & & \\
        & \ddots & \ddots & \ddots & & \\
        & & \mu & -(\lambda+\mu) & \lambda & \\
        & & & \ddots & \ddots & \ddots
    \end{array} \right).
\end{equation}
This is the same continuous-time Markov chain as in a birth-death process. 

Let\begin{equation}
   \mathcal{N} =  \left( \begin{array}{cccccccccc}
                   0 & 1& & & & \\
                     & 0 & 1 & & & \\
                    & & \ddots & \ddots & & \\
                    & & & 0 & 1 & \\
                    & & & & \ddots & \ddots
                \end{array} \right) \text{and}\ \mathcal{I} =  \left( \begin{array}{cccccccccc}
                   1 & & & & & \\
                     & 1 & & & & \\
                    & & \ddots & & & \\
                    & & & 1 &  & \\
                    & & & & \ddots & \end{array} \right).
\end{equation} 
We can write the real part and the imaginary part of $Q = \mathbf L + \mathrm{i} \mathbf H$ as $\mathbf L=(\mu+\lambda)(-\mathcal{I} + \frac{1}{2}\mathcal{N}+ \frac{1}{2}\mathcal{N}^\dagger$) and $\mathbf H=\frac{1}{2}(\mu-\lambda)(\mathrm{i}\mathcal{N} -\mathrm{i}\mathcal{N}^\dagger$). Since all the infinite-dimensional matrices $\mathcal{I},\mathcal{N}$ and $\mathcal{N}^\dagger$ are sparse-1 matrices with elements have the uniform upper bound, thus $\mathbf L$ and $\mathbf H$ are bounded operators, which satisfy condition \ref{condition2} and the self-adjointness condition in \cref{adjoint_assum}.

For the M/M/c model as the generalization of the M/M/1 model, which considers only a single server, the transition rate matrix $Q$ is
\begin{equation}
    \left( \begin{array}{cccccccccc}
        -(\lambda+\mu) & \lambda & & & & \\
        2\mu & -(\lambda+2\mu) & \lambda & & & \\
        & \ddots & \ddots & \ddots & & \\
        & & c_n\mu & -(\lambda+c_n\mu) & \lambda & \\
        & & & \ddots & \ddots & \ddots 
    \end{array} \right),
\end{equation}
where $c_n=n$ is $n\leq c$, $c_n=c$ if $n>c$. Thus, the real and imaginary parts of $Q$ is:
\begin{equation}
   \vb{L} =  \left( \begin{array}{cccccccccc}
        -(\lambda+\mu) & \frac{\lambda + 2\mu}{2} & & & & \\
        \frac{\lambda + 2\mu}{2} & -(\lambda+2\mu) & \frac{\lambda + 3\mu}{2} & & & \\
        & \ddots & \ddots & \ddots & & \\
        & & \frac{\lambda + c_n\mu}{2} & -(\lambda+c_n\mu) & \frac{\lambda + c_{n+1}\mu}{2} & \\
        & & & \ddots & \ddots & \ddots 
    \end{array} \right)
\end{equation} 
and 
\begin{equation}
\vb {H} = \frac{1}{2\mathrm{i}} \left( \begin{array}{cccccccccc}
        0 & \lambda - 2\mu & & & & \\
        2\mu -\lambda & 0 & \lambda - 3\mu & & & \\
        & \ddots & \ddots & \ddots & & \\
        & & c_n\mu -\lambda & 0 & \lambda - c_{n+1}\mu& \\
        & & & \ddots & \ddots & \ddots 
    \end{array} \right).
\end{equation} 
Similar to the reason in the M/M/1 model above, $\mathbf L$ and $\mathbf H$ are bounded operators, satisfying condition \ref{condition2} and \cref{adjoint_assum}. Therefore, \cref{main_thm} holds for all the M/M/c queue models if $\mathbf L\succ0$.

For the M/M/1 and M/M/c models, since $\mathbf L$ is always a bounded operator, it follows directly from \cref{THM:Inf-LCHS-G} that the number of unitaries in the Inf-LCHS theorem using Gaussian quadrature has the scaling $N\in \mc O(T\norm{\mathbf L} (\log 1/\varepsilon)^{1+1/\beta})$. Here $\norm{\mathbf L} \le 2(\lambda + \mu)$ for the M/M/1 model and $\norm{\mathbf L} \le 2(\lambda + c\mu)$ for the M/M/1 model. Similarly, under the settings of the Inf-LCHS-MC scheme (\cref{THM:Inf-LCHS-MC}), the number $N$ scaling with $\mc O([\log(1/\varepsilon)]^{2/\beta}/\varepsilon^2)$ is independent with parameters $\lambda$ and $\mu$.

\subsection{Schrödinger Equations with Complex Potentials}\label{apply: schro}

We consider the real-space Schrödinger equation in the semi-classical regime~\cite{JMS11,LasserLubich2020,JLL21, FangVilanova2023, BornsWeilFang2025}
\begin{align}\label{equ: schrodinger equ}
    \mathrm{i}\hbar\frac{\partial u(x,t)}{\partial t} = \left[-\frac{\hbar^2}{2}\Delta+V(x,t)\right] u(x,t), \qquad (x,t)\in \mathbb{R}^d\times\mathbb{R}^+,
\end{align}
where $\hbar \ll 1$. There are various notable examples, such as the Bohn-Oppenheimer approximation and the Erhenfest dynamics, whereas $\hbar^2$ is the ratio between the electron and nuclei mass. When dividing $\hbar$ on both sides and taking $\hbar \mapsto 0$, both the kinetic operator $\frac{\hbar}{2}\Delta$ and potential operator $\frac{1}{\hbar}V(x,t)$ behave as unbounded operators.

To simulate the dynamics in the finite domain $D\subset\mathbb{R}^d$ rather than the entire space, we suitably develop Artificial Boundary Conditions (ABCs) to prevent boundary reflections. For instance, the Complex Absorbing Potential (CAP) method~\cite{VibokBalintKurti1992, Child1991, MugaPalaoNavarroEtAl2004}, we replace the real potential by a complex one $V_R+\mathrm{i}V_I$, resulting in the equation
\begin{align}
    \mathrm{i}\hbar\frac{\partial u(x,t)}{\partial t} = \left[-\frac{\hbar^2}{2}\Delta+V_R(x,t)+\mathrm{i}V_I(x)\right] u(x,t),\qquad (x,t)\in D\times\mathbb{R}^+.
\end{align}
This formulation can be naturally generalized to multi-particle Schrödinger equations.
Other ABCs, such as the perfectly matched layer (PML) method and Dirichlet-to-Neumann (DtN) method, can be formulated similarly~\cite{JLLY23,JLLY24}.

In regards to the Inf-LCHS theorem, we take $A = \mathbf L + \mathrm{i}\mathbf H$ with $\mathbf L = - \frac{1}{\hbar}V_I(x)$ and $\mathbf H = -\frac{\hbar}{2}\Delta + \frac{1}{\hbar}V_R(x,t)$. Additionally, we define $Y(t)=\int_0^t \mathbf L(s)\mathrm{~d}s =- \frac{t}{\hbar}V_I(x)$ and $X(t)=-\int_0^t \mathbf H(s)\mathrm{~d}s = -\frac{t\hbar}{2}\Delta + \frac{1}{\hbar}\int_0^t V_R(x,s)\mathrm{~d}s$. 
Assume that the real part $V_I(x)$ is a continuous function with compact support for $x\in\mathbb{R}^d$, and that $V_R(x,t)$ is an arbitrary potential function. The compact support condition is a standard assumption in CAP methods~\cite{MugaPalaoNavarroEtAl2004}.  In this case, $V_I(x)$ is a bounded multiplication operator, which satisfies condition \ref{condition2} in \cref{main_thm}. 
Therefore, for a broad class of $V_I(x)$ with compact support, the explicit solution to the initial value problem, where $u(x,0) = u_0$, can be expressed according to the Inf-LCHS theorem as follows:
\begin{align}
    u(x,t) = \Bigl( \int_{\mathbb{R}} \frac{f(k)}{1-\mathrm{i} k} \mathcal{T} e^{-\mathrm{i} \int_0^t(k \mathbf L(s)+\mathbf H(s)) \mathrm{d} s} \mathrm{~d} k \Bigr) ~u_0,
\end{align}
as long as the other assumptions of Inf-LCHS theorem are satisfied. To facilitate the reader's understanding, we demonstrate the satisfiability of all assumptions for a specific Schrödinger equation in \cref{APD 2}.

However, compact support of the potential is not the sole condition under which the Inf-LCHS theorem holds. We can also consider the case outlined in condition \ref{condition1} of \cref{main_thm}. For example, consider the scenario where the real part $V_I(x)$  is a linear unbounded potential, and $V_R(x,t)$ is a polynomial in $x$. This type of potential is often considered within the context of parity-time ($\mc P\mc T$)-symmetric quantum theory, whose discovery and implementation can give rise to entirely new and unexpected phenomena. As a result, it has generated significant interest in the study of non-Hermitian systems both experimentally and theoretically~\cite{ElGanainyMakrisKhajavikhanEtAl2018}. In this case, we can show that the operator $Z(X(t),Y(t))$ exists, and that \cref{main_thm} holds provided that the other necessary conditions are satisfied.
    We first consider the one-dimensional case, where $x\in\mathbb{R}$. Suppose the solution is $u(x, t)$ and $V_I(x) = x+c$. Then, we compute the following: 
    \begin{align*}
        [\Delta, \mathcal{M}_{V_I(x)}] u = &\partial_x(V_I^\prime(x) u + V_I(x) u^\prime) - V_I(x)u^{\prime\prime}\\
        = &V_I^{\prime\prime}(x)u + 2V_I^\prime(x)\partial_x u\\
        = & (\mathcal{M}_{V_I^{\prime\prime}(x)} + 2V_I^\prime(x)\partial_x)u\\
        = & 2\partial_x u.
    \end{align*}
    Thus, the condition simplifies to condition 1 in \cref{prop: commut partialx}. By \cref{lemma: first order partial commutant}, $Z(X(t), Y(t))$ can be expressed as a finite sum of nested Lie brackets. The calculation in the higher-dimensional case follows similarly, with the computation being carried out for each dimension independently.

To estimate the number of unitaries needed in the Inf-LCHS theorem, or the number of Schrödinger equations needed in the decomposition of the original differential equation, we apply the results in \cref{sec:sample}. For Inf-LCHS-Gaussian (\cref{THM:Inf-LCHS-G}), we note that $\mathbf L = \mathcal{M}_{V_I(x)}$ is time-independent and we have
\begin{equation}
    \norm{\mathbf L_{\rm grid}} = \norm{\mathrm{diag} (\{V_{I}(\xi_i)\}_{1\le i \le N_{\rm grid}})} \le \max_{1\le i\le N_{\rm grid}}\abs{V_I(\xi_i)} =: \norm{V_I}_{\Xi} <\infty ,
\end{equation}
since $D \subset \mathbb{R}^d$ is a finite domain. Here, $\mathbf L_{\text{grid}}$ denotes the matrix obtained by discretizing the spatial variable using the set of grid points $\Xi = \{ \xi_i \}_{1 \leq i \leq N_{\text{grid}}}$. According to \cref{THM:Inf-LCHS-G}, the number of Schrödinger equations required for the Inf-LCHS decomposition scales with
\begin{equation}
    N = \mc O \left(T \norm{V_I}_{\Xi} \left(\log \frac1\varepsilon\right)^{1+1/\beta}\right).
\end{equation}
Similarly, for the Inf-LCHS-MC scheme (\cref{THM:Inf-LCHS-MC}), the total number of Schrödinger equations for simulating the PDE scales with 
\begin{equation}
    N = \mc O \left(
    \frac1{\varepsilon^2} \left(\log \frac1\varepsilon\right)^{2/\beta}
    \right),
\end{equation}
which has worse asymptotic complexity than the Inf-LCHS-Gaussian scheme. However, it eliminates the dependency on the multiplication operator $V_I(x)$ or the simulation time $T$, which can be poorly bounded or highly non-trivial to estimate.

\subsection{Lindblad Equations}

The Gorini--Kossakowski--Sudarshan--Lindblad (GKSL) quantum master equation is widely used for modeling Markovian open quantum systems, where the system is weakly coupled to a sufficiently large bath. Moreover, the Lindblad equation has emerged as a powerful algorithmic tool for applications such as quantum Gibbs state~\cite{ChenKastoryanoBrandãoEtAl2023,ChenKastoryanoGilyén2023, BrunnerCoopmansMatosEtAl2024} and ground state preparation~\cite{DingChenLin2024,LiZhanLin2024}, attracting significant attention in both quantum algorithms and quantum many-body physics in recent years.

Mathematically, the Lindblad equation has a remarkably broad range of applicability, including infinite-dimensional cases. In its most general form, the Lindbladian can be defined as the generator of a \emph{quantum dynamical semigroup} in general. Specifically, let $\mathcal{H}$ be a Hilbert space, $\mathcal{A}\subset \mathcal{B}(\mathcal{H})$ a von Neumann algebra, and $\omega$ be a faithful $\sigma$-weakly continuous strictly  positive tracial state on $\mathcal{A}$. Then we can define the Hilbert-Schmidt inner product as $\langle a,b\rangle:=\omega(b^\ast a)$ for $a,b\in \mathcal{A}$, and $(\mathcal{A},\langle\cdot,\cdot\rangle)$ is a Hilbert space.  The GKSL theorem claims that a bounded linear map $\mathcal{L}\in \mathcal{B}(\mathcal{A})$ is the generator of a uniformly continuous, conservative quantum Feller semigroup composed of normal maps if and only if it takes the following form: 
\begin{equation}
    \mathcal{L}(a) = -\mathrm{i}[H,a] -\frac{1}{2}(L^\ast La -2L^\ast(a\otimes I) L + aL^\ast L ),\quad \forall a\in\mathcal{A},
\end{equation}
Here $H=H^\ast$ is a self-adjoint operator in $\mathcal{A}$ and $L\in\mathcal{B}(\mathcal{H},\mathcal{H}\otimes \mathcal{K})$ for some Hilbert space $\mathcal{K}$. We denote 
\begin{equation}
    \mathbf H(a) = [H,a],\quad \mathbf L(a) = \frac{1}{2}(L^\ast L a- 2L^\ast (a\otimes I) L + aL^\ast L), 
\end{equation} then $\mathbf H$ and $\mathbf L$ are both self-adjoint operators in $\mathcal{B}(\mathcal{A})$ with respect to the Hilbert-Schmidt inner product. And self-adjointness condition in \cref{adjoint_assum} also holds, as both $\mathbf H$ and $\mathbf L$ are bounded. To see $\mathbf L\succ 0$, we write
$
    \mathbf L (a) =  \frac{1}{2}(\Phi (\mathbf{1}) a +a \Phi (\mathbf 1)-2\Phi(a))
$
according to the proof of Theorem 4.55 in~\cite{Bahns}, where $\Phi:\mc A\to \mc A, a\mapsto L^\ast (a\otimes I) L$
is completely positive and normal. Since $\Phi^\dagger$ is a positive mapping and $\omega$ is a strictly positive tracial state, we note that $ \omega(a^\ast \mathbf L(a)) = \omega(\Phi^\dagger(a^\ast a+aa^\ast)/2) > 0$. By $\mathcal{L} = -(\mathbf L + \mathrm{i}\mathbf H)$ and condition \ref{condition2} in \cref{main_thm}, we obtain the solution formula of Lindblad equation with the initial state $a_0\in \mc A$:
\begin{equation}
    a_t  =\int_{\mathbb{R}} \frac{f(k)}{1-\mathrm{i}k} e^{-\mathrm{i} (k\mathbf L + \mathbf H) t}[a_0]\dd k .
\end{equation}

In practice, numerical methods for simulating Lindblad dynamics, such as the quantum jump method and the quantum state diffusion method, also incorporate non-Hermitian dynamics to model Markovian open quantum systems. Additionally, there is considerable interest in the generalized GKSL equation with unbounded generators~\cite{SiemonHolevoWerner2017}. Notably, our Inf-LCHS theorem can be extended to address these cases as well.

For the complexity regarding the number of Schr\"odinger equations, by \cref{THM:Inf-LCHS-G} and \cref{THM:Inf-LCHS-MC}, the number $N$ scales $\mc O(T\norm{\mathbf L} (\log 1/\varepsilon)^{1+1/\beta})$ with $\norm{\mathbf L} \le 2\norm{ L}^2$ and $\mc O([\log (1/\varepsilon)]^{2/\beta} /\varepsilon^2)$ for Inf-LCHS-Gaussian and Inf-LCHS-MC methods, respectively.

\subsection{Black Hole Thermal Field Equations}

Quantum field theory provides a consistent framework for relativistic quantum systems by replacing wavefunctions with field operators \( \hat{\phi}(x) \), satisfying  
\[
[\hat{\phi}(x), \hat{\phi}(y)] = 0, \quad (x - y)^2 < 0.
\]
This enforces causality while treating particles as excitations of fields, making quantum field theory essential for relativistic quantum systems. Relativistic quantum field theory (RQFT) solves causality violations in single-particle quantum mechanics when the amplitude \( A  = \langle x | e^{-\mathrm i H t} | x = 0 \rangle\) is  nonzero for a particle to travel outside its forward light cone (i.e., move to a space-like separated point where \( |x| > t \)).

Our Inf-LCHS theorem also holds for the black hole thermal field equations in RQFT. The evolution of black hole states under Hawking radiation has been a subject of extensive investigation, particularly concerning the implications for information loss. In scenarios where the black hole evaporation process is non-unitary, the time-evolution operator takes the form 
    \begin{equation}
        U(t) = e^{(-\mathrm{i} H-\gamma) t},
    \end{equation}
    where \( H \) is the standard Hermitian Hamiltonian governing the unitary dynamics, and \( \gamma \) is a positive decay rate associated with information loss. This modification leads to an exponential suppression of quantum coherence, thereby altering the entropy evolution of the black hole.

The quantum state \( \psi(t) \) describing the black hole evolution in a non-unitary scenario satisfies the modified Schrödinger equation with the initial vector $\psi(0)=\psi_0$:
\begin{equation}
    \dv t \psi(t) = -\mathrm{i} H \psi(t) - \gamma \psi(t) = -\mathcal{L}\psi(t).
\end{equation}
The additional term \( -\gamma \psi(t) \) introduces an irreversible decay of quantum information, which may be interpreted as an effective description of black hole information loss. 

We rewrite the PDE operator in the form $\mathcal{L} = \mathrm{i}H + \gamma I$. Since $\mathcal{L} = \mathbf L + \mathrm{i}\mathbf H$, we have $\mathbf L= \gamma I$, $\mathbf H = H$. If $H$ is self-adjoint, it is easy to verify that \cref{adjoint_assum} holds. In this case, since $\mathbf L= \gamma I$ is a bounded operator satisfying condition \ref{condition2}, and $\gamma$ is positive, by \cref{main_thm}, we have 
\begin{equation}
   \psi(t) = U(t)\psi_0 =\int_{\mathbb{R}} \frac{f(k)}{1-\mathrm{i}k} e^{-\mathrm{i}(k\gamma I + H) t}\dd k\psi_0.
\end{equation}

The complexity associated with the number of Schrödinger equations can be characterized using \cref{THM:Inf-LCHS-G} and \cref{THM:Inf-LCHS-MC}, which yield that the number $N$ scales as  
$\mathcal{O}(T\gamma (\log 1/\varepsilon)^{1+1/\beta})$  
for the Inf-LCHS-Gaussian method and  
$\mathcal{O}([\log (1/\varepsilon)]^{2/\beta} /\varepsilon^2)$  
for the Inf-LCHS-MC method.

\section{Discussion and Outlook}\label{sec:discussion}

In this work, we derive the infinite-dimensional extension of the linear combination of Hamiltonian simulation (Inf-LCHS) formula for the time-evolution of unbounded or infinite-dimensional operators, establishing its applicability to PDE problems. For the ODE case, the formula reduces to the well-known finite-dimensional linear combination of Hamiltonian simulation (LCHS) formula. For infinite-dimensional systems, we establish the Inf-LCHS theorem to approximate PDEs with Schrödinger equations with more subtle assumptions and different sampling methods. This advancement raises several open questions: 

Condition \ref{condition1} is introduced primarily to ensure the validity of the BCH formula for unbounded operators. However, if more refined conditions can be identified that specify the BCH formula's scope and limitations, this would provide a deeper understanding of the Inf-LCHS theorem and offer a clearer insight into its potential scope.

In addition, condition \ref{condition2} restricts the applicability of the Inf-LCHS theorem. Is it possible to relax the boundedness condition in \ref{condition2} with other assumptions? This question arises naturally in the challenges of expressing the real and imaginary parts of nonlinear PDE operators. Specifically, can the Inf-LCHS approach be extended to nonlinear PDEs, and if so, what modifications would enable handling the nonlinearities in the time-evolution operators?

The LCHS framework achieves the near-optimal quantum algorithm for ODEs~\cite{ACL23}, but numerical algorithms for PDEs should demand extra assumptions, e.g., the boundedness of $\|L(t)\|$ in \cref{THM:Inf-LCHS-G}. The Inf-LCHS-MC approach in \cref{THM:Inf-LCHS-MC} can handle unbounded operators and relax the dependence of norm $\|L(t)\|$ and evolution time $T$, but it takes higher asymptotic complexity than the Inf-LCHS-Gaussian approach. What are the optimal strategies and computational limits of solving PDEs on quantum computers, and how can the problem-specific structures of certain PDEs reduce the quantum cost? Can we find other novel applications of propagating time-evolution operators via Inf-LCHS in mathematics and physics?

\begin{acknowledgments}
We appreciate insightful discussions with Dong An, Di Fang, Daniel Galviz, Lin Lin, Arthur Jaffe, and Jinsong Wu. JPL acknowledges support from Tsinghua University and Beijing Institute of Mathematical Sciences and Applications. ZWL was supported by Beijing Natural Science Foundation
Key Program (Grant No. Z220002) and by NKPs (Grant no. 2020YFA0713000). RDL and ZWL were supported by BMSTC and ACZSP (Grant No. Z221100002722017). HEL was partially supported by the National
Natural Science Foundation of China (Grant No.
223B1011).
\end{acknowledgments}

\bibliography{ref.bib}

\onecolumngrid
\clearpage

\appendix
\begin{center}
    \textbf{Supplementary Materials}
\end{center}

\section{Preliminaries}\label{APD 1}

\begin{definition}
Let \( H \) be a Hilbert space, the linear operator \( A: D(A) \subset H \to H \). The domain \( D\left(A^{\dagger}\right) \) is the largest domain on which the adjoint operator \( A^{\dagger} \) can be defined. Specifically,
$$
D\left(A^{\dagger}\right) = \{ g \in H \mid \text{the map } f \in D(A) \mapsto (A f, g) \text{ is a continuous linear functional} \}.
$$
\( A \)  is said to be \textit{symmetric} if
$$
(A f, g) = (f, A g), \quad \forall f, g \in D(A).
$$
\( A \) is said to be \textit{self-adjoint} if
$$
D(A) = D\left(A^{\dagger}\right) \quad \text{and} \quad (A f, g) = (f, A g), \quad \forall f, g \in D(A).
$$
Finally, \( A \) is said to be \textit{essentially self-adjoint} if its closure is self-adjoint, that is,
$$
\overline{A} = A^\dagger,
$$
where \( \overline{A} \) denotes the closure of the operator \( A \).
\end{definition}

\begin{definition}
The domain \( D\left(A^{\dagger}\right) \) is the largest domain on which the adjoint operator \( A^{\dagger} \) can be defined. Specifically,
$$
D\left(A^{\dagger}\right) = \{ g \in H \mid \text{the map } f \in D(A) \mapsto (A f, g) \text{ is a continuous linear functional} \}.
$$
\end{definition}
\begin{definition}
 An operator $B$ is called $A$-bounded if $D(B) \supset D(A)$.
It is obvious, but important to note, that any bounded operator $B \in \mathcal{L}(\mathcal{H})$ is $A$-bounded for any linear operator $A$.
\end{definition}

\section{Example Demonstrating the Validity of Inf-LCHS Assumptions}\label{APD 2}

In this section, to help the reader understand the purpose and verification methods behind the assumptions of the main theorem, we provide a specific example of the Schrödinger equation to illustrate that the assumptions introduced in \cref{main_thm}, which are intended to ensure mathematical rigor, are both reasonable and can indeed be satisfied.

The standard definitions of the differential operators $\Delta$, $\mathrm{i}\frac{\partial}{\partial x}$, and the elliptic operator $-\sum_{i,j}\partial_i\left(a^{ij}(x,t)\partial_j f\right) + c(x,t)$ reveal their symmetry properties. However, in the context of the Inf-LCHS  formula, we require the self-adjoint decomposition of the time evolution operator. It is important to note that proving a symmetric operator is self-adjoint is a non-trivial task, and such a property does not hold in all cases. Therefore, the real and imaginary parts of the differential operators discussed in this paper are considered to act on function spaces where they can be extended to self-adjoint operators by taking their closures. For instance, $\Delta$ is an essentially self-adjoint operator on $\mathcal{C}_c^\infty(\mathbb{R}^n)$, and it is self-adjoint on $H^2(\mathbb{R}^d)$.

Based on the above considerations, the real and imaginary parts of the differential operators discussed in this paper are regarded as self-adjoint operators on appropriate spaces by taking their self-adjoint extensions. In the following, we demonstrate that the self-adjoint assumption in \cref{adjoint_assum} is rigorously satisfied for the example Schrödinger equation.

In this section, we consider the real-space Schrödinger equation in the semi-classical regime
    \begin{align}
        \mathrm{i}\hbar\frac{\partial u(x,t)}{\partial t} = \left[ -\frac{\hbar^2}{2} \Delta + V_R(x,t) + \mathrm{i} V_I(x) \right] u(x,t), \qquad (x,t) \in \mathbb{R}^3 \times \mathbb{R}^+,
    \end{align}
    where $V_R$ is of the form $V_R(x,t) = V_1(t) V_2(x)$, with $V_1(t)$ being continuous on $[0, T]$, the potential $V_2(x)$ in the Schrödinger equation is constant, i.e., $V_2(x) = V_2$, and $V_I(x) \in H^2(\mathbb{R}^3) \cap L^\infty(\mathbb{R}^3) \subset L^2(\mathbb{R}^3) + L^\infty(\mathbb{R}^3)$. We define the operator $A(t) = \mathbf{L}(t) + \mathrm{i} \mathbf{H}(t)$, where $\mathbf{L}(t) = - \frac{1}{\hbar} V_I(x)$ and $\mathbf{H}(t) = - \frac{\hbar}{2} \Delta + \frac{1}{\hbar} V_R(x,t)$. Additionally, we define 
    \[
    Y(t) = \int_0^t \mathbf{L}(s) \, \mathrm{d}s = -\frac{t}{\hbar} V_I(x), \quad X(t) = -\mathrm{i} \int_0^t \mathbf{H}(s) \, \mathrm{d}s = -\mathrm{i} \left( -\frac{t \hbar}{2} \Delta + \frac{1}{\hbar} \int_0^t V_R(x,s) \, \mathrm{d}s \right).
    \]
    
First, we demonstrate that \cref{adjoint_assum} is satisfied by utilizing the following result in~\cite{Hislop1995IntroductionTS}.
\begin{lemma}\label{Apd: shcodinger adjoint}
Let $V \in L^2\left(\mathbb{R}^3\right)+L^{\infty}\left(\mathbb{R}^3\right)$ and be real. Then the operator $H=-\Delta + V$, defined on $D(\Delta) =  H^2(\mathbb{R}^3)$, is self-adjoint.
\end{lemma}

\begin{proposition}\label{apd prop: schro adjoint assum}
    For each fixed $t \in [0, T]$, the operators $\mathbf{L}$, $\mathbf{H}$, $Y(t)$, $\mathrm{i} X(t)$, and $kY(t) + \mathrm{i} X(t)$, for all $k \in \mathbb{R}$, are self-adjoint on $H^2(\mathbb{R}^3)$. Consequently, the Schrödinger equation of this form satisfies the assumptions stated in \cref{adjoint_assum}.
\end{proposition}

\begin{proof}
    Since $V_I(x) \in H^2(\mathbb{R}^3) \cap L^\infty(\mathbb{R}^3)$, the operator $\mathbf{L} = - \frac{1}{\hbar} V_I(x)$ is a bounded self-adjoint operator on $H^2(\mathbb{R}^3)$.
    
    If $V_R, V_I \in L^2(\mathbb{R}^3) + L^\infty(\mathbb{R}^3)$, then for each fixed $t \in [0, T]$, the term $- \frac{t}{\hbar} V_I(x)$, as well as $\frac{1}{\hbar} \int_0^t V_R(x,s) \, \mathrm{d}s = \left( \frac{1}{\hbar} \int_0^t V_1(s) \, \mathrm{d}s \right) V_2(x)$, both belong to $L^2(\mathbb{R}^3) + L^\infty(\mathbb{R}^3)$. 

    For any $k \in \mathbb{R}$, the operator $kY(t) + \mathrm{i} X(t)$ becomes 
    \[
    -\frac{t \hbar}{2} \Delta + \left( \frac{1}{\hbar} \int_0^t V_1(s) \, \mathrm{d}s \right) V_2(x) - \frac{t k}{\hbar} V_I(x).
    \]
    Clearly, the potential 
    \[
    \left( \frac{1}{\hbar} \int_0^t V_1(s) \, \mathrm{d}s \right) V_2(x) - \frac{t k}{\hbar} V_I(x)
    \]
    belongs to $L^2(\mathbb{R}^3) + L^\infty(\mathbb{R}^3)$ as well. By \cref{Apd: shcodinger adjoint}, it follows that the operators $\mathbf{L}$, $\mathbf{H}$, $Y(t)$, $\mathrm{i} X(t)$, and $kY(t) + \mathrm{i} X(t)$, for all $k \in \mathbb{R}$, are self-adjoint on $H^2(\mathbb{R}^3)$.
\end{proof}

\begin{proposition}\label{time order domain}
   The Hamiltonian defined by this Schrödinger equation has a non-degenerate initial vector domain $X_0= H^2(\mathbb{R}^3)$ , which satisfies \cref{assum: domain}.
\end{proposition}
\begin{proof}
    Under this assumption, we can express the time-ordered exponential as follows:
    \begin{align}
        \mathcal{T} e^{-\mathrm{i} \int_0^t (k \mathbf{L}(s) + \mathbf{H}(s)) \, \mathrm{d}s} \psi_0 
        &= \mathcal{T} e^{\mathrm{i} \left( \frac{t \hbar}{2} \Delta + \frac{t k}{\hbar} V_I(x) - \left( \frac{1}{\hbar} \int_0^t V_1(s) \, \mathrm{d}s \right) V_2 \right)} \psi_0 \\
        &= \mathcal{T} e^{-\mathrm{i} \left( \frac{1}{\hbar} \int_0^t V_1(s) \, \mathrm{d}s \right) V_2} e^{\mathrm{i} \left( \frac{t \hbar}{2} \Delta + \frac{t k}{\hbar} V_I(x) \right)} \psi_0,
    \end{align}
    where the second equality follows from the commutation relation 
    \[
    \left[ -\mathrm{i} \left( \frac{1}{\hbar} \int_0^t V_1(s) \, \mathrm{d}s \right) V_2, \, \mathrm{i} \left( \frac{t \hbar}{2} \Delta + \frac{t k}{\hbar} V_I(x) \right) \right] = 0.
    \]

    Let $X_0 = \mathcal{D}(\Delta) = H^2(\mathbb{R}^3)$, and for all $\psi_0 \in X_0$, by Stone's theorem, we have
    \[
    \psi_t^1 := e^{\mathrm{i} \left( \frac{t \hbar}{2} \Delta + \frac{t k}{\hbar} V_I(x) \right)} \psi_0 \in \mathcal{D}\left( \frac{t \hbar}{2} \Delta + \frac{t k}{\hbar} V_I(x) \right) = \mathcal{D}(\Delta) = X_0.
    \]
    Since the potential $V_2(x)$ is constant, the time-ordered exponential factor
    \[
    \mathcal{T} e^{-\mathrm{i} \left( \frac{1}{\hbar} \int_0^t V_1(s) \, \mathrm{d}s \right) V_2}
    \]
    is simply a scalar in $X_0$. Therefore, the total wavefunction becomes
    \[
    \psi_t = \mathcal{T} e^{-\mathrm{i} \left( \frac{1}{\hbar} \int_0^t V_1(s) \, \mathrm{d}s \right) V_2} \psi_t^1 \in H^2(\mathbb{R}^3) = \mathcal{D}\left( -\mathrm{i} \left( \frac{k}{\hbar} V_I(x) - \frac{\hbar}{2} \Delta + \frac{1}{\hbar} V_R(x,t) \right) \right) = \mathcal{D}\left( -\mathrm{i} \left( k \mathbf{L}(t) + \mathbf{H}(t)\right) \right).
    \]

    Thus, for all $\psi_0 \in X_0$, we have $\mathcal{T} e^{-\mathrm{i} \int_0^t (k L(s) + H(s)) \, \mathrm{d}s} \psi_0 \in X_0 = H^2(\mathbb{R}^3)$. Therefore, $X_0$ satisfies the condition stated in \cref{assum: domain}.
\end{proof}
 
\begin{proposition}
    $A(t) \in \operatorname{SOT}([0,T], L(X))$.
\end{proposition}
\begin{proof}
    For any $t$, $t_0\in[0, T]$, 
    \begin{align}
        \|A(t)\psi_0 - A(t_0)\psi_0\| &\leq \|\mathbf{L}(t)\psi_0 - \mathbf{L}(t_0)\psi_0\| + \|\mathbf{H}(t)\psi_0  -\mathbf{H}(t_0)\psi_0\|\\
        &= |t-t_0|\cdot\|\frac{1}{\hbar} V_I(x)\psi_0\| + |V_1(t)-V_1(t_0)|\cdot\|\frac{1}{\hbar}V_2\psi_0\|.
    \end{align}
    Thus, $\lim_{t \to t_0} \|A(t)\psi_0 - A(t_0)\psi_0\| = 0$ by the continuity of $V_1(t)$.
\end{proof}

\section{Supplementary Refinements to the Proof of the Main Theorem}\label{apd: supply proof}
In \cref{main_thm_formula1}, we compute the time derivative of $\psi_t$ by moving the operator $\frac{d}{dt}$ from outside the integral to inside. This operation is rigorous, as we demonstrate in the following proposition.

\begin{proposition}[Rigorousness of the calculation in \cref{main_thm_formula1}]\label{APD: order change}
In the calculation process of \cref{main_thm_formula1}, the order of differentiation with respect to $t$ and integration over $k$ can be interchanged.
\end{proposition}

\begin{proof}
    To establish that the derivative operator $\frac{\mathrm{d}}{\mathrm{d} t}$ and the integral $\int_\mathbb{R} \cdot \, \dd k$ in the first line of \cref{main_thm_formula1} can be swapped, we invoke the Leibniz integral rule for differentiation. It suffices to show that the vector-valued functions
    $$\mathcal{F}(k, t) = \frac{f(k)}{1-\mathrm{i} k} \mathcal{T} \exp\left(-\mathrm{i} \int_0^t (k L(s) + H(s)) \, \mathrm{d} s \right) \psi_0$$
    and
    $$\mathcal{F}_t(k, t) = \frac{f(k)}{1-\mathrm{i} k}  \left(-\mathrm{i} k L(t) - \mathrm{i} H(t)\right) \mathcal{T} \exp\left(-\mathrm{i} \int_0^t (k L(s) + H(s)) \, \mathrm{d} s \right) \psi_0$$
    are continuous in $(k,t) \in \mathbb{R} \times [0, T]$. Since $f(k)$ is continuous by \cref{assum: func}, it remains to verify the continuity of the functions $g(k, t) = \mathcal{T} \exp\left(-\mathrm{i} \int_0^t (k L(s) + H(s)) \, \mathrm{d} s \right) \psi_0$ and $g_t(k, t) = \left(-\mathrm{i} k L(t) - \mathrm{i} H(t)\right) \mathcal{T} \exp\left(-\mathrm{i} \int_0^t (k L(s) + H(s)) \, \mathrm{d} s \right) \psi_0$.

    Since $A(t) \in  \operatorname{SOT}([0,T],L(X))$, we can immediately verify that $A^\dagger(t)$, $L(t)$, and $H(t)$ are all in $ \operatorname{SOT}([0,T],L(X))$ by definition. Therefore, it suffices to show that the function 
    $$\mathcal{T} \exp\left(-\mathrm{i} \int_0^t (k L(s) + H(s)) \, \mathrm{d} s \right) \psi_0 \in C([0,T], L(X)).$$
    
    By the definition of the time-ordering operator, for all $\psi_0 \in X_0$ and $t \in [0, T]$, we have
    \begin{align}
        \left\| \frac{\mathrm{d}}{\mathrm{d} t} \mathcal{T} \exp\left(-\mathrm{i} \int_0^t (k L(s) + H(s)) \, \mathrm{d} s \right) \psi_0 \right\| &= \left\| (k L(t) + H(t)) \mathcal{T} \exp\left(-\mathrm{i} \int_0^t (k L(s) + H(s)) \, \mathrm{d} s \right) \psi_0 \right\| \nonumber \\
        &\leq \infty.
    \end{align}
    Therefore, for all $t, t_0 \in [0, T]$, we obtain
    \begin{align}
        \lim_{t \to t_0} \left\| \mathcal{T} \exp\left(-\mathrm{i} \int_0^t (k L(s) + H(s)) \, \mathrm{d} s \right) \psi_0 - \mathcal{T} \exp\left(-\mathrm{i} \int_0^{t_0} (k L(s) + H(s)) \, \mathrm{d} s \right) \psi_0 \right\| = 0.
    \end{align}
    This establishes the continuity of the functions $\mathcal{F}(k, t)$ and $\mathcal{F}_t(k, t)$ in $t$.
\end{proof}

Thus, by the proof of \cref{APD: order change}, the functions $\mathcal{F}(k, t)$ and $\mathcal{F}_t(k, t)$ are continuous in $(k,t) \in \mathbb{R} \times [0, T]$, ensuring that the integrands before and after the application of the derivative operator $\frac{d}{dt}$ are integrable.

Let $X$ be a Hilbert space and $\mathcal{D}$ be an open subset of $\mathbb{C}$. 

\begin{definition}
     Let $E, F$ be Banach spaces and let $T: E \rightarrow F$ be a linear operator with the domain $\mathcal{D}(T) \subset E$. $T$ is a closed operator if given any sequence $\left\{x_n\right\}_{n \in \mathbb{N}} \subset$ $\mathcal{D}(T)$ such that $\lim x_n=x \in E$ and $\lim T\left(x_n\right)=y \in F$, then $x \in \mathcal{D}(T)$ and $T(x)=$ $y$.
\end{definition}

\begin{definition}
Let $\mathbf{F}: \mathcal{D} \rightarrow \mathscr{L}(X)$. Suppose for all $z\in \mathcal{D}$, $\mathbf{F}(z)$ is the unbounded closed linear operator defined on the domain $\mathscr{D}(\mathbf{F})\subseteq X$.
\begin{enumerate}[label=(\alph*)]
    \item $\mathbf{F}$ is analytic at $z_0 \in \mathcal{D}$ if
$$
\mathbf{F}^{\prime}\left(z_0\right)u=\lim _{z \rightarrow z_0} \frac{\mathbf{F}(z)-\mathbf{F}\left(z_0\right)}{z-z_0}u
$$
exists (in the norm on $X$) for all $u\in X$.

 \item $\mathbf{F}$ is weakly analytic at $z_0 \in \mathcal{D}$ if
$$
\lim _{z \rightarrow z_0} \frac{\langle\mathbf{F}(z)(u), v\rangle-\langle\mathbf{F}(z_0)(u), v\rangle}{z-z_0}
$$
exists for each $u\in \mathscr{D}(\mathbf{F}),\ v\in X$. In other words, $\mathbf{F}$ is weakly analytic at $z_0 \in \mathcal{D}$ if, for each $u\in \mathscr{D}(\mathbf{F}),\ v\in X$, the function $h: \mathcal{D} \rightarrow \mathbb{C}$ defined by $h(z)=\langle\mathbf{F}(z)(u), v\rangle$ is analytic at $z_0$.
\end{enumerate}
\end{definition}

We will prove that Cauchy's integral theorem still holds in Hilbert space.
\begin{proposition}[Cauchy's integral theorem for unbounded operator]\label{Cauchy's integral theorem}
Let $\mathscr{L}(X)$ be the set of all the unbounded linear operator on $X$, and $F: \mathcal{D}\to \mathscr{L}(X)$, the following arguments hold:
\begin{enumerate}[label=(\alph*)]
\item $\mathbf{F}(z)$ is analytic on $\mathcal{D}$, then $\mathbf{F}(z)$ is weakly analytic on $\mathcal{D}$.
\item 
If  $\mathcal{D}$ is an open subset of $\mathbb{C}, \mathbf{F}$ : $\mathcal{D}\rightarrow \mathscr{L}(X)$ is an analytic function, and $\gamma_1, \ldots, \gamma_m$ are closed rectifiable curves in $\mathcal{D}$ such that $\sum_{j=1}^m n\left(\gamma_j ; a\right)=0$ for all $a$ in $\mathbb{C} \backslash \mathcal{D}$, then $\sum_{j=1}^m \int_{\gamma_j} \mathbf{F}=0$.
\end{enumerate}
\end{proposition}
\begin{proof}
(a). For each $z_0 \in \mathcal{D}, u\in \mathscr{D}(\mathbf{F})$ and $v\in X$,
\begin{align*}
   & \left|\frac{\langle\mathbf{F}(z)(u), v\rangle-\langle\mathbf{F}(z_0)(u), v\rangle}{z-z_0}-\langle\mathbf{F}^{\prime}(z_0)(u), v\rangle\right|=\left\langle\left(\frac{\mathbf{F}(z) - \mathbf{F}(z_0)}{z-z_0}- \mathbf{F}^{\prime}(z_0)\right)(u), v\right\rangle \\  
  \leq&\|v\| \cdot \left\|\left(\frac{\mathbf{F}(z)-\mathbf{F}\left(z_0\right)}{z-z_0}-\mathbf{F}^{\prime}\left(z_0\right)\right)(u)\right\|
\end{align*}

(b). For all $u\in \mathscr{D}(\mathbf{F}),\ v\in X$, then $$\left\langle\left(\sum_{j=1}^m \int_{\gamma_j} \mathbf{F}(z)dz\right)(u), v\right\rangle=\sum_{j=1}^m \int_{\gamma_j}\left\langle \mathbf{F}(z)(u), v\right\rangle dz=0$$ by (a) and the scalar-valued version of Cauchy's Theorem. Hence $\sum_{j=1}^m \int_{\gamma_j} \mathbf{F}(z)dz=0$ on $\mathscr{D}(\mathbf{F})$.
\end{proof}

\begin{lemma}[BCH formula]
If the operator in right hand side of \cref{BCH} exists, then
    \begin{align}\label{BCH}
        \log (\exp X \exp Y) = & X+Y+\frac{1}{2}[X, Y]+\frac{1}{12}([X,[X, Y]]]+[Y,[Y, X]]) + \cdots \\\notag
        = &\sum_{n=1}^{\infty} \frac{(-1)^{n-1}}{n} \sum_{r_1+s_1>0}\cdots\sum_{r_n+s_n>0} \frac{\left[X^{r_1} Y^{s_1} X^{r_2} Y^{s_2} \cdots X^{r_n} Y^{s_n}\right]}{\left(\sum_{j=1}^n\left(r_j+s_j\right)\right) \cdot \prod_{i=1}^n r_{i}!s_{i}!}
    \end{align}
where the sum is performed over all nonnegative values of $s_i$ and $r_i$, and the following notation has been used:
$$
\left[X^{r_1} Y^{s_1} \cdots X^{r_n} Y^{s_n}\right]=[\underbrace{X,[X, \cdots[X}_{r_1},[\underbrace{Y,[Y, \cdots[Y}_{s_1}, \cdots[\underbrace{X,[X, \cdots[X}_{r_n},[\underbrace{Y,[Y, \cdots Y]}_{s_n}] \cdots]]
$$
with the understanding that $[X]:=X$.
\end{lemma}

\begin{lemma}\label{lemma:Trotter}
    Let $H$ be a Hilbert space, not necessarily separable. If $A$ and $B$ are self-adjoint, the semigroups $e^{tB}$ and $e^{t(\mathrm{i}A + B)}$ are contraction and $A+ iB$ is densely defined, then for each $t \in \mathbb{R}$ and for each $\psi \in H$,
$$
e^{\overline{\mathrm{i}tA + tB}} \psi = e^{\mathrm{i}t \overline{(A - \mathrm{i} B)}}\psi=\lim _{n \rightarrow \infty}\left(\left(e^{\mathrm{i} t A / n} e^{t B / n}\right)^n \psi\right).
$$
\end{lemma}
\begin{proof}
    The claim is immediate for $t=0$, and we prove the claim for $t>0$; it is straightforward to obtain the claim for $t<0$ using the truth of the claim for $t>0$. Let $D=\mathscr{D}(A-\mathrm{i}B)=\mathscr{D}(A) \cap \mathscr{D}(B)$. Because $A+ \mathrm{i}B$ is densely defined, so $A-\mathrm{i}B = (A+\mathrm{i}B)^\dagger$ is closed, and thus the linear space $D$ with the norm $\|\phi\|_{A-\mathrm{i}B}=\|\phi\|+\|(A-\mathrm{i}B) \phi\|$ is a Banach space. Because $D$ is a Banach space, the uniform boundedness principle tells us that if $\Gamma$ is a collection of bounded linear maps $D \rightarrow H$ and if for each $\phi \in D$ the set $\{\gamma \phi: \gamma \in \Gamma\}$ is bounded in $H$, then the set $\{\|\gamma\|: \gamma \in \Gamma\}$ is bounded, i.e. there is some $C$ such that $\|\gamma \phi\| \leq C\|\phi\|_{A-\mathrm{i}B}$ for all $\gamma \in \Gamma$ and all $\phi \in D$.

For $s \in \mathbb{R}$, let $S_s=e^{\mathrm{i} s(A-\mathrm{i}B)}, V_s=e^{\mathrm{i} s A}, W_s=e^{ s B}, U_s=V_s W_s$, which each belong to $\mathscr{B}(H)$. For $n \geq 1$,
$$
\sum_{j=0}^{n-1} U_{t / n}^{j}\left(S_{t / n}-U_{t / n}\right) S_{t / n}^{n-j-1}=S_{t / n}^n-U_{t / n}^n=S_t-U_{t / n}^n, 
$$
 because semigroups $e^{tB}$ and $e^{t(\mathrm{i}A + B)}$ are contraction and $e^{\mathrm{i}tA}$ is unitary and also using the fact that $S_{t / n}^{n-j-1}=S_{t(1-\frac{j+1}{n})}$,
for $\xi \in H$ we have
$$
\begin{aligned}
\left\|\left(S_t-U_{t / n}^n\right) \xi\right\| & =\left\|\sum_{j=0}^{n-1} U_{t / n}^j\left(S_{t / n}-U_{t / n}\right) S_{t / n}^{n-j-1} \xi\right\| \\
& \leq \sum_{j=0}^{n-1}\left\|\left(S_{t / n}-U_{t / n}\right) S_{t / n}^{n-j-1} \xi\right\| \\
& =\sum_{j=0}^{n-1}\left\|\left(S_{t / n}-U_{t / n}\right) S_{t(1-\frac{j+1}{n})} \xi\right\| \\
& \leq \sum_{j=0}^{n-1} \sup _{0 \leq s \leq t}\left\|\left(S_{t / n}-U_{t / n}\right) S_s \xi\right\|
\end{aligned}
$$

That is,
\begin{equation}\label{APD_Lemma:0}
    \left\|\left(S_t-U_{t / n}^n\right) \xi\right\| \leq n \sup _{0 \leq s \leq t}\left\|\left(S_{t / n}-U_{t / n}\right) S_s \xi\right\|, \quad \xi \in H, \quad n \geq 1
\end{equation}

Let $\phi \in D$. On the one hand, because $\mathrm{i}(A-\mathrm{i}B)$ is the infinitesimal generator of $\left\{S_s: s \in \mathbb{R}\right\}$, we have
\begin{equation}\label{APD_Lemma:1}
    \frac{S_s-I}{s} \phi \rightarrow \mathrm{i}(A-\mathrm{i}B) \phi, \quad s \downarrow 0 .
\end{equation}

On the other hand, for $s \neq 0$ we have, because an infinitesimal generator of a one-parameter group commutes with each element of the one-parameter group,
$$
V_s(B \phi)+V_s\left(\frac{W_s-I}{s}-B\right) \phi+\frac{V_s-I}{s} \phi=\frac{U_s-I}{s} \phi
$$
and as $V_s$ converges strongly to $I$ as $s \downarrow 0$ and as $B$ is the infinitesimal generator of the one-parameter group $\left\{W_s: s \in \mathbb{R}\right\}$ and $i A$ is the infinitesimal generator of the one-parameter group $\left\{V_s: s \in \mathbb{R}\right\}$,
$$
V_s(B \phi)+V_s\left(\frac{W_s-I}{s}-B\right) \phi+\frac{V_s-I}{s} \phi \rightarrow B \phi+\mathrm{i} A \phi
$$
as $s \downarrow 0$, i.e.
\begin{equation}\label{APD_Lemma:2}
    \frac{U_s-I}{s} \phi \rightarrow \mathrm{i}(A-\mathrm{i}B) \phi, \quad s \downarrow 0 .
\end{equation}

Using \cref{APD_Lemma:1} and \cref{APD_Lemma:2}, we get that for each $\phi \in D$,
$$
\frac{S_s-U_s}{s} \phi \rightarrow 0, \quad s \downarrow 0 .
$$

Therefore, for each $\phi \in D$, with $s=t / n$ we have
$$
\frac{n}{t}\left(S_{t / n}-U_{t / n}\right) \phi \rightarrow 0, \quad n \rightarrow \infty
$$
equivalently ( $t$ is fixed for this whole theorem),

\begin{equation}\label{APD_Lemma:3}
    \lim _{n \rightarrow \infty}\left\|n\left(S_{t / n}-U_{t / n}\right) \phi\right\|=0, \quad \phi \in D
\end{equation}

For each $n \geq 1$, define $\gamma_n: D \rightarrow H$ by $\gamma_n=n\left(S_{t / n}-U_{t / n}\right)$. Each $\gamma_n$ is a linear map, and for $\phi \in D$,
$$
\left\|\gamma_n \phi\right\| \leq n\left\|S_{t / n} \phi\right\|+n\left\|U_{t / n} \phi\right\| \leq n\|\phi\|+n\|\phi\| \leq 2 n\|\phi\|_{A-iB},
$$
showing that each $\gamma_n$ is a bounded linear map $D \rightarrow H$, where $D$ is a Banach space with the norm $\|\phi\|_{A-\mathrm{i}B}=\|\phi\|+\|(A-\mathrm{i}B) \phi\|$. Moreover, \cref{APD_Lemma:3} shows that for each $\phi \in D$, there is some $C_\phi$ such that
$$
\left\|\gamma_n \phi\right\| \leq C_\phi, \quad n \geq 1 .
$$

Then applying the uniform boundedness principle, we get that there is some $C>0$ such that for all $n \geq 1$ and for all $\phi \in D$,
$$
\left\|\gamma_n \phi\right\| \leq C\|\phi\|_{A-\mathrm{i}B},
$$
i.e.
\begin{equation}\label{APD_Lemma:4}
    \left\|n\left(S_{t / n}-U_{t / n}\right) \phi\right\| \leq C\|\phi\|_{A-\mathrm{i}B}, \quad n \geq 1, \quad \phi \in D .
\end{equation}

Let $K$ be a compact subset of $D$, where $D$ is a Banach space with the norm $\|\phi\|_{A-\mathrm{i}B}=\|\phi\|+\|(A-\mathrm{i}B) \phi\|$. Then $K$ is totally bounded, so for any $\epsilon>0$, there are $\phi_1, \ldots, \phi_M \in K$ such that $K \subset \bigcup_{m=1}^M B_{\epsilon / C}\left(\phi_m\right)$. By \cref{APD_Lemma:3}, for each $m$, $1 \leq m \leq M$, there is some $n_m$ such that when $n \geq n_m$,
$$
\left\|n\left(S_{t / n}-U_{t / n}\right) \phi_m\right\| \leq \epsilon
$$

Let $N=\max \left\{n_1, \ldots, n_M\right\}$. For $n \geq N$ and for $\phi \in K$, there is some $m$ for which $\left\|\phi-\phi_m\right\|_{A-\mathrm{i}B}<\frac{\epsilon}{C}$, and using \cref{APD_Lemma:4}, as $\phi-\phi_m \in K$, we get
$$
\begin{aligned}
\left\|n\left(S_{t / n}-U_{t / n}\right) \phi\right\| & \leq\left\|n\left(S_{t / n}-U_{t / n}\right)\left(\phi-\phi_m\right)\right\|+\left\|n\left(S_{t / n}-U_{t / n}\right) \phi_m\right\| \\
& \leq C\left\|\phi-\phi_m\right\|_{A-\mathrm{i}B}+\epsilon \\
& <\epsilon+\epsilon
\end{aligned}
$$

This shows that any compact subset $K$ of $D$ and $\epsilon>0$, there is some $n_\epsilon$ such that if $n \geq n_\epsilon$ and $\phi \in K$, then
\begin{equation}\label{APD_Lemma:5}
    \left\|n\left(S_{t / n}-U_{t / n}\right) \phi\right\|<\epsilon
\end{equation}

Let $\phi \in D$, let $s_0 \in \mathbb{R}$, and let $\epsilon>0$. Because $s \mapsto S_s$ is strongly continuous $\mathbb{R} \rightarrow \mathscr{B}(H)$, there is some $\delta_1>0$ such that when $\left|s-s_0\right|<\delta_1$, $\left\|S_s \phi-S_{s_0} \phi\right\|<\epsilon$, and there is some $\delta_2>0$ such that when $\left|s-s_0\right|<\delta_2$, $\left\|S_s(A-\mathrm{i}B) \phi-S_{s_0}(A-\mathrm{i}B) \phi\right\|<\epsilon$, and hence with $\delta=\min \left\{\delta_1, \delta_2\right\}$, when $\left|s-s_0\right|<\delta$ we have
$$
\begin{aligned}
\left\|S_s \phi-S_{s_0} \phi\right\|_{A-\mathrm{i}B} & =\left\|S_s \phi-S_{s_0} \phi\right\|+\left\|(A-\mathrm{i}B)\left(S_s \phi-S_{s_0} \phi\right)\right\| \\
& \left.=\left\|S_s \phi-S_{s_0} \phi\right\|+\| S_s(A-\mathrm{i}B) \phi-S_{s_0}(A-\mathrm{i}B) \phi\right) \| \\
& <\epsilon+\epsilon
\end{aligned}
$$
showing that $s \mapsto S_s \phi$ is continuous $\mathbb{R} \rightarrow D$. Therefore $\left\{S_s \phi: 0 \leq s \leq t\right\}$ is a compact subset of $D$, so applying \cref{APD_Lemma:5} we get that for any $\epsilon>0$, there is some $n_\epsilon$ such that if $n \geq n_\epsilon$ and $0 \leq s \leq t$, then
$$
\left\|n\left(S_{t / n}-U_{t / n}\right) S_s \phi\right\|<\epsilon,
$$
and therefore if $n \geq n_\epsilon$ then
\begin{equation}\label{APD_Lemma:6}
    \sup _{0 \leq s \leq t}\left\|n\left(S_{t / n}-U_{t / n}\right) S_s \phi\right\| \leq \epsilon
\end{equation}

Finally, let $\epsilon>0$. Since $D$ is dense in $H$, so for all $\psi\in H$, there is some $\phi \in D$ such that $\|\phi-\psi\|<\epsilon$. For $n \geq 1$,
$$
\begin{aligned}
\left\|\left(S_t-U_{t / n}^n\right) \psi\right\| & \leq\left\|\left(S_t-U_{t / n}^n\right)(\psi-\phi)\right\|+\left\|\left(S_t-U_{t / n}^n\right) \phi\right\| \\
& \leq 2\|\psi-\phi\|+\left\|\left(S_t-U_{t / n}^n\right) \phi\right\| \\
& <\epsilon+\left\|\left(S_t-U_{t / n}^n\right) \phi\right\|
\end{aligned}
$$

Using \cref{APD_Lemma:0} with $\xi=\phi$ and then using \cref{APD_Lemma:6}, there is some $n_\epsilon$ such that when $n \geq n_\epsilon$,
$$
\left\|\left(S_t-U_{t / n}^n\right) \phi\right\| \leq n \sup _{0 \leq s \leq t}\left\|\left(S_{t / n}-U_{t / n}\right) S_s \phi\right\| \leq \epsilon .
$$

Therefore for $n \geq n_\epsilon$,
$$
\left\|\left(S_t-U_{t / n}^n\right) \psi\right\|<2 \epsilon,
$$
proving the claim.
\end{proof}

\begin{lemma}\label{lemma_main}
     Let $f(z)$ be a function satisfying the conditions in \cref{main_thm}, $Y(t)=\int_0^t L(s)\mathrm{~d}s$ and $X(t)=-\mathrm{i}\int_0^t H(s)\mathrm{~d}s$. Suppose $L(t) \succeq \lambda_0>0$ for a positive number $\lambda_0$ and all $t \in[0, T]$, and operator $-z L(s)+\mathrm{i} H(s)$ is densely defined. If operators $L(t)$ and $H(t)$ satisfy one of the two conditions in \cref{main_thm}, then
\begin{equation}
    \mathcal{P} \int_{\mathbb{R}} f(k) \mathcal{T} e^{-\mathrm{i} \int_0^t(k L(s)+H(s)) \mathrm{d} s} \mathrm{~d} k=0
\end{equation}
Here $\mathcal{P}$ denotes the Cauchy principal value.
\end{lemma}
\begin{proof}
    Using the substitution $\omega=\mathrm{i} k$, we have
\begin{align}\label{int_trans}
    \mathcal{P} \int_{\mathbb{R}} f(k) \mathcal{T} e^{-\mathrm{i} \int_0^t(k L(s)+H(s)) \mathrm{d} s} \mathrm{~d} k=-\mathrm{i} \lim _{R \rightarrow \infty} \int_{-\mathrm{i} R}^{\mathrm{i} R} f(-\mathrm{i} \omega) \mathcal{T} e^{\int_0^t(-\omega L(s)-\mathrm{i} H(s)) \mathrm{d} s} \mathrm{~d} \omega
\end{align}

Let us choose a path $\gamma_C=\left\{\omega=R e^{\mathrm{i} \theta}: \theta \in[-\pi / 2, \pi / 2]\right\}$ and consider the contour $[-\mathrm{i} R, \mathrm{i} R] \cup \gamma_C$. According to the analyticity condition, $f(-\mathrm{i} \omega)$ is analytic on the right half plane $\{\omega: \operatorname{Re}(\omega)>0\}$ and continuous on the right half plane plus the imaginary axis. Since operators $X(t)$ and $Y(t)$ satisfies one of the following two conditions, Cauchy's integral theorem holds by \cref{lemma_limit} or \cref{lemma:limit_trotter}, we have
\begin{equation}\label{fomula: analytic function}
    \int_{[-\mathrm{i} R, \mathrm{i} R] \cup \gamma_C} f(-\mathrm{i} \omega) \mathcal{T} e^{\int_0^t(-\omega L(s)-\mathrm{i} H(s)) \mathrm{d} s} \mathrm{~d} \omega=0
\end{equation}

Plugging this back into \cref{int_trans}, we obtain
\begin{align*}
    \mathcal{P} \int_{\mathbb{R}} f(k) \mathcal{T} e^{-\mathrm{i} \int_0^t(k L(s)+H(s)) \mathrm{d} s} \mathrm{~d} k= & \mathrm{i} \lim _{R \rightarrow \infty} \int_{\gamma_C} f(-\mathrm{i} \omega) \mathcal{T} e^{\int_0^t(-\omega L(s)-\mathrm{i} H(s)) \mathrm{d} s} \mathrm{~d} \omega
\end{align*}

Changing the variable from $\omega$ to $\theta$ gives

\begin{align}\label{int_lemma_fomula4}
\mathcal{P} \int_{\mathbb{R}} f(k) \mathcal{T} e^{-\mathrm{i} \int_0^t(k L(s)+H(s)) \mathrm{d} s} \mathrm{~d} k \notag=& \lim _{R \rightarrow \infty} \int_{-\pi / 2}^{\pi / 2} R e^{\mathrm{i} \theta} f\left(-\mathrm{i} R e^{\mathrm{i} \theta}\right) \mathcal{T} e^{\int_0^t\left(-R e^{\mathrm{i} \theta} L(s)-\mathrm{i} H(s)\right) \mathrm{d} s} \mathrm{~d} \theta \notag\\
=& \lim _{R \rightarrow \infty} \int_{I \cup J} R e^{\mathrm{i} \theta} f\left(-\mathrm{i} R e^{\mathrm{i} \theta}\right) \mathcal{T} e^{\int_0^t\left(-R e^{\mathrm{i} \theta} L(s)-\mathrm{i} H(s)\right) \mathrm{d} s} \mathrm{~d} \theta
\end{align}

where $I=\left[-\pi / 2+\theta_0, \pi / 2-\theta_0\right], J=[-\pi / 2, \pi / 2] \backslash I$, and
$$
\theta_0=\min \left\{\frac{1}{R^{1-\alpha / 2}}, \frac{\pi}{4}\right\}
$$

Here $\alpha$ is the parameter in the decay condition. By the spectral theorem of the adjoint operators~\cite{Hall2013QuantumTF} and Trotter decomposition in \cref{lemma:Trotter} which holds for the same derivation in \cref{lemma:limit_trotter},  we can bound the integral over $I$ by the inequality
\begin{align}
    \left\|\mathcal{T} e^{\int_0^t\left(-R e^{\mathrm{i} \theta} L(s)-\mathrm{i} H(s)\right) \mathrm{d} s}\right\| = \notag &\left\|\mathcal{T} \lim_{n\to \infty}( e^\frac{-R\cos  \theta\int_0^t L(s)\mathrm{d} s}{n}e^\frac{-\mathrm{i}\int_0^t\left(R \sin  \theta  L(s)+ H(s)\right) \mathrm{d} s}{n})^n\right\| \\\notag
    \leq & \lim_{n\to \infty}\| e^\frac{-R\cos  \theta\int_0^t L(s)\mathrm{d} s}{n}\|^n \\\notag
 \leq & e^{-t \lambda_0 R \cos \theta} \\
 \leq & e^{-\frac{2}{\pi} t \lambda_0 R \theta_0}
\end{align}

where in the last inequality we use $\cos \theta \geq \cos \left(\pi / 2-\theta_0\right)=\sin \theta_0 \geq(2 / \pi) \theta_0$. 
Then
\begin{align}\label{int_lemma_fomula2}
    \left\|\int_I R e^{\mathrm{i} \theta} f\left(-\mathrm{i} R e^{\mathrm{i} \theta}\right) \mathcal{T} e^{\int_0^t\left(-R e^{\mathrm{i} \theta} L(s)-\mathrm{i} H(s)\right) \mathrm{d} s} \mathrm{~d} \theta\right\| \leq  \int_I R\left|f\left(-\mathrm{i} R e^{\mathrm{i} \theta}\right)\right| e^{-\frac{2}{\pi} t \lambda_0 R \theta_0} \mathrm{~d} \theta
\end{align}

For sufficiently large $R$, using the decay condition and noticing that $\left(-\mathrm{i} R e^{\mathrm{i} \theta}\right)$ is in the lower half plane for $\theta \in[-\pi / 2, \pi / 2]$, we have
\begin{align}\label{int_lemma_fomula3}
    \left|f\left(-\mathrm{i} R e^{\mathrm{i} \theta}\right)\right| \leq \frac{\widetilde{C}}{R^\alpha}
\end{align}

Then \cref{int_lemma_fomula2} becomes
\begin{align*}
    \left\|\int_I R e^{\mathrm{i} \theta} f\left(-\mathrm{i} R e^{\mathrm{i} \theta}\right) \mathcal{T} e^{\int_0^t\left(-R e^{\mathrm{i} \theta} L(s)-\mathrm{i} H(s)\right) \mathrm{d} s} \mathrm{~d} \theta\right\| 
    \leq \pi \widetilde{C} R^{1-\alpha} e^{-\frac{2}{\pi} t \lambda_0 R \theta_0} 
    \leq  \pi \widetilde{C} R^{1-\alpha} \max \left\{e^{-\frac{2}{\pi} t \lambda_0 R^{\alpha / 2}}, e^{-\frac{1}{2} t \lambda_0 R}\right\}
\end{align*}

which vanishes as $R \rightarrow \infty$. For the integral over $J$, we may again use the bound of $f$ in \cref{int_lemma_fomula3} and $\left\|\mathcal{T} e^{\int_0^t\left(-R e^{\mathrm{i} \theta} L(s)-\mathrm{i} H(s)\right) \mathrm{d} s}\right\| \leq e^{-t \lambda_0 R \cos \theta} \leq 1$ for any $\theta \in[-\pi / 2, \pi / 2]$. Then we obtain
\begin{align*}
     \left\|\int_J R e^{\mathrm{i} \theta} f\left(-\mathrm{i} R e^{\mathrm{i} \theta}\right) \mathcal{T} e^{\int_0^t\left(-R e^{\mathrm{i} \theta} L(s)-\mathrm{i} H(s)\right) \mathrm{d} s} \mathrm{~d} \theta\right\| 
    \leq \int_J R\left|f\left(-\mathrm{i} R e^{\mathrm{i} \theta}\right)\right| \mathrm{d} \theta 
    \leq \widetilde{C} R^{1-\alpha} \cdot 2 \theta_0 
    \leq  2 \widetilde{C} R^{-\alpha / 2}
\end{align*}

which also vanishes as $R \rightarrow \infty$. Therefore, \cref{int_lemma_fomula4} becomes
$$
\mathcal{P} \int_{\mathbb{R}} f(k) \mathcal{T} e^{-\mathrm{i} \int_0^t(k L(s)+H(s)) \mathrm{d} s} \mathrm{~d} k=0
$$
\end{proof}

\end{document}